\documentclass[11pt]{article}
\usepackage{mathpazo}

\usepackage[dvipsnames,svgnames,x11names]{xcolor}
\usepackage{amsthm}
\usepackage{amsmath}
\usepackage{amssymb}
\usepackage[colorlinks]{hyperref}
\usepackage{xspace}
\usepackage{relsize}
\usepackage{tcolorbox}

\usepackage{tabu}

\usepackage{geometry}

\newcommand\p{\mbox{\bf P}\xspace}
\newcommand\np{\mbox{\bf NP}\xspace}

\renewcommand{\inf}{\mathbf{Inf}}

\newtheorem{theorem}{Theorem}[section]

\newtheorem{definition}[theorem]{Definition}
\newtheorem{lemma}[theorem]{Lemma}
\newtheorem{remark}[theorem]{Remark}
\newtheorem{proposition}[theorem]{Proposition}
\newtheorem{corollary}[theorem]{Corollary}
\newtheorem{claim}[theorem]{Claim}
\newtheorem{fact}[theorem]{Fact}

\newcommand\B{\mathcal{B}}

\newcommand\calH{\mathcal{H}}

\newcommand\calV{\mathcal{V}}

\newcommand\calI{\mathcal{I}}
\newcommand\calU{\mathcal{U}}

\newcommand\E{\mathop{\mathbf{E}}}

\newcommand{\eps}{\varepsilon}
\hypersetup{linkcolor=blue,citecolor=blue,filecolor=blue,urlcolor=blue}

\newcommand{\lref}[2][]{\hyperref[#2]{#1~\ref*{#2}}}
\renewcommand{\eqref}[1]{\hyperref[#1]{(\ref*{#1})}}

\usepackage{etoolbox}
\makeatletter
\patchcmd{\@bibitem}{\ignorespaces}{\label{bib-#1}\ignorespaces}{}{}
\makeatother


\renewcommand{\le}{\leqslant}

\renewcommand{\Pr}{\mathbf{Pr}}

\newcommand{\C}{\mathbb{C}}
\newcommand{\GL}{\mathsf{GL}}

\newcommand{\V}[1]{\boldsymbol{#1}}
\newcommand{\trivrep}{\{\mathbf{1}\}}
\newcommand{\gp}{\text{\relsize{-3} $\bullet$}\hspace{2pt}} %groupproduct
\newcommand{\cp}{\cdot} %complexproduct

\newcommand{\irr}{\mathsf{Irrep}}
\newcommand{\trace}{\mathsf{tr}}
\newcommand{\comp}{\overline}
\newcommand{\End}[1]{\mathsf{End}\hspace{1pt} #1}
\newcommand{\conv}{*}
\newcommand{\ctranspose}[1]{{#1}^\star}
\newcommand{\hsnorm}[1]{\| #1 \|_{\mathsf{HS}}}
\newcommand{\dimgeqi}[2]{\dim_{\geq #2}({#1})}

\newcommand{\commutator}[2]{%
	{\left[%
	\begin{array}{
			@{}
			>{\centering $\displaystyle}p{0.8em}<{$}
			@{,}
			>{\centering $\displaystyle}p{0.8em}<{$}
			@{}
		}
		#1 & #2
	\end{array}%
	\right]%
	}
}

\newcommand{\low}{\mathsf{low}}
\newcommand{\high}{\mathsf{high}}

\newcommand{\abnote}[1]{{\color{red} (#1)}}

\newcommand{\matrixb}[1]{{{ \left[ {#1}\right]}}}

\renewcommand{\geq}{\geqslant}
\renewcommand{\leq}{\leqslant}

\def\showreplacements{0}

\definecolor{timberwolf}{rgb}{0.76, 0.74, 0.72}
\definecolor{darkspringgreen}{rgb}{0.09, 0.45, 0.27}

\ifnum\showreplacements=1

\else

\fi

\usepackage{tikz}
\usetikzlibrary{patterns,arrows.meta}
\allowdisplaybreaks

\title{Optimal Inapproximability of Satisfiable $k$-LIN over Non-Abelian Groups}

\author{Amey Bhangale\thanks{Department of Computer Science and Engineering, University of California, Riverside, CA, USA. 
		Email: {\tt ameyrb@ucr.edu}}
\and 
Subhash Khot\thanks{Department of Computer Science, Courant Institute of Mathematical Sciences, New York University, NY, USA. 
	Email: {\tt khot@cs.nyu.edu }}
}

\date{}

\newcommand{\ii}{p}
\newcommand{\jj}{q}
\newcommand{\kk}{r}

\newcommand{\LIN}{LIN}

\newcommand{\N}{\mathbb{N}}
\newcommand{\val}{val}
\newcommand{\LC}{{\sc Label-Cover}}
\newcommand{\TSAT}{3\text{-}SAT}

\begin{document}
\maketitle	

\begin{abstract}
A seminal result of H\r{a}stad~\cite{Hastad2001}  shows that for $k\geq 3$, it is NP-hard to find an assignment that satisfies $\frac{1}{|G|}+\eps$ fraction of the constraints of a given $k$-LIN instance over an abelian group, even if there is an assignment that satisfies $(1-\eps)$ fraction of the constraints, for any constant $\eps>0$.  Engebretsen, Holmerin, and Russell~\cite{EngebretsenHR} later showed that the same hardness result holds for $k$-LIN instances over any finite non-abelian group.

Unlike the abelian case, where we can efficiently find a solution if the instance is satisfiable, in the non-abelian case, it is NP-complete to decide if a given system of linear equations is satisfiable or not, as shown by Goldmann and Russell~\cite{RG99}.  

Surprisingly, for certain non-abelian groups $G$, given a satisfiable $k$-LIN instance over $G$, one can in fact do better than just outputting a random assignment using a simple but clever algorithm. The approximation factor achieved by this algorithm varies with the underlying group. In this paper, we show that this algorithm is {\em optimal} by proving a tight hardness of approximation of a satisfiable $k$-LIN instance over {\em any} non-abelian $G$, assuming $\p\neq \np$. As a corollary, we also get $3$-query probabilistically checkable proofs with perfect completeness over large alphabets with improved soundness.

Our proof crucially uses structural properties captured by Gowers’ notion of {\em quasi-randomness} for non-abelian groups~\cite{Gowers08}.
\end{abstract}

\section{Introduction}
Constraint satisfaction problems (CSPs) are the most fundamental problems in computer science. One of the simplest such CSPs that we know how to solve is a system of $k$-\LIN\ equations over some abelian group. More generally, an instance of Max-$k$-\LIN\ over a group $G = (G, \gp)$, not necessarily abelian,  consists of a set of variables $x_1, x_2, \ldots, x_n$ and a set of constraints $C_1, C_2, \ldots, C_m$. Each $C_i$ is a linear equation involving $k$ variables, for example $a_1\gp x_{i_1}\gp a_2\gp x_{i_2}\gp \ldots \gp a_k \gp x_{i_k} = b$, for some group elements $a_1, a_2, \ldots, a_k, b\in G$. The task is to find an assignment to the variables that satisfies as many constraints as possible.

For {\em any} abelian group $G$, if there is a perfect solution to the given Max-$k$-\LIN\ instance over $G$, then it can be found efficiently~\cite{RG99} in polynomial time by solving modular linear systems. A given instance is {\em almost} satisfiable if there exists an assignment that satisfies $(1-\eps)$-fraction of the constraints for a small constant $\eps>0$. If the given instance of Max-$k$-\LIN\ over an abelian group $G$ is almost satisfiable, then H\r{a}stad~\cite{Hastad2001} showed that it is $\np$-hard to even find an assignment that satisfies a $\frac{1}{|G|}+\eps$ fraction of the constraints for every constant $\eps>0$. In other words, one cannot do significantly better than just outputting a random assignment. 

The situation changes completely if the instance is a set of linear equations over a non-abelian group. In this case, Goldmann and Russell~\cite{RG99} showed that even the problem of deciding if a given instance is satisfiable or not is $\np$-complete for {\em every} non-abelian group.

\paragraph{An algorithm (folklore):} It turns out that for certain groups, one can do much better than assigning a random assignment when the instance is satisfiable. Given an instance $\phi$ over $G$, consider an instance $\phi'$ over $H  = G/\commutator{G}{G}$ where $\commutator{G}{G}$ is a commutator subgroup of $G$, i.e., the subgroup generated by the elements $\{g^{-1}h^{-1}gh \mid  g, h\in G\}$. The instance $\phi'$ is the same as $\phi$ except that all group constants are replaced by their equivalence class in $G/\commutator{G}{G}$. The important property of this quotient group $H$ is that it is an abelian group. Since $\phi$ has a satisfying assignment over $G$, $\phi'$ has a satisfying assignment over $H$. Hence, we can find the satisfying assignment $\sigma$ of $\phi'$ in polynomial time. The solution $\sigma$ is an assignment of cosets of $\commutator{G}{G}$ to the variables. We construct a random assignment to $\phi$ such that for every variable $x$, we select a random group element from $\sigma(x)$ and assign it to $x$. It is easy to see that each constraint is satisfied with a probability equal to $\frac{1}{\left|\commutator{G}{G}\right|}$. Thus, this gives an assignment that satisfies $\frac{1}{\left|\commutator{G}{G}\right|}$ fraction of the constraints in expectation. Therefore, if $\commutator{G}{G}$ is a proper subgroup of $G$, then we get an algorithm that does better than the random assignment threshold.\\

If the instance is almost satisfiable, then it is not clear how to modify the above algorithm  for almost satisfiable instances to get better than $\frac{1}{|G|}$ approximation. In fact, for almost satisfiable instances over any non-abelian group,  Engebretsen, Holmerin, and Russell~\cite{EngebretsenHR} showed that it is \np-hard to do better than outputting a random assignment.

This leaves an intriguing question of finding the correct approximation threshold for satisfiable instances over non-abelian groups. In this paper, we show that the above-described algorithm for satisfiable instances over non-abelian groups is the best one can hope for. More specifically, we prove the following theorem.
\begin{theorem}
	\label{thm:mainnphardness}
	For any constant $\eps>0$, given a {\em satisfiable} instance of Max-$3$-\LIN\  over a finite non-abelian group $G$, it is \np-hard to find an assignment that satisfies $\frac{1}{\left| \commutator{G}{G}\right|} + \eps$ fraction of the constraints.
\end{theorem}

The theorem can be extended to Max-$k$-\LIN\ for any $k\geq 3$ to imply a similar hardness result for all Max-$k$-\LIN\ problems over $G$.\\

If $G$ is a {\em perfect group}, i.e., $|\commutator{G}{G}| = |G|$, then Theorem~\ref{thm:mainnphardness} implies an NP-hardness of approximating satisfiable  Max-$3$-CSP instances over an alphabet of size $q=|G|$ to within a factor of $\frac{1}{q}+\eps$, for every constant $\eps>0$. As a direct consequence, we get improved soundness of $3$-query probabilistically checkable proofs (PCPs) with perfect completeness over large alphabets. Since PCPs are not the main focus of this paper, we refer interested readers to the book by Arora and Barak~\cite[Chapter 18]{AroraB09} to see the relation between PCPs and CSPs.

\begin{corollary}
For infinitely many $q\in \mathbb{Z}^+$, any language in \np is decided by a nonadaptive PCP with answers from
a domain of size $q$ that queries three positions in the proof, has perfect completeness and
soundness $\frac{1}{q}+\eps$ for any constant $\eps>0$. 
\end{corollary}

This improves a result by Engebretsen and Holmerin~\cite{EngebretsenH05} where they constructed PCPs with soundness $\frac{1}{q} + \frac{1}{q^2}+ \eps$, and also a result by Tang~\cite{Tang09} in which they showed a conditional result with soundness $\frac{1}{q} + \frac{1}{q^2} - \frac{1}{q^3} +\eps$, for any constant $\eps>0$.\footnote{The theorem in~\cite{EngebretsenH05} holds for every $q \geq 3$, and the theorem in~\cite{Tang09} holds for every $q\geq 4$.  Our theorem holds for $q$ such that there is a perfect group of order $q$.} 
\subsection{Techniques}

We assume some familiarity with the Fourier analysis of functions over abelian groups (for instance, Chapter 8 of Ryan O'Donnell's book~\cite{O14}). Throughout the section, $\eps >0$ is an arbitrarily small constant and $\delta(\eps)>0$ decays with $\eps$. We only discuss $3$-\LIN\ here, however the argument is similar for $k$-\LIN\ in general.

Given the body of work on hardness of approximation -- where many reductions begin from the Label Cover problem (see Definition~\ref{def:label-cover}) -- a central ingredient in these reductions is the use of \emph{dictatorship tests}. A function $f: G^n \rightarrow G$ is a dictator function if it depends only on one variable in the following way: $f(x_1, x_2, 
	\ldots, x_n) = x_i$ for some $i\in [n]$. On the other hand, we have functions which are {\em far } from dictator functions. 
	
	To understand a notion of distance from a dictator function, which is useful for a reduction to work, we first consider a few reasonable notions and finally arrive at the right one.  Consider a scalar representation $h:G^n\rightarrow \mathbb{C}$ of a $G$-valued function $f$ by having a one-to-one map $\tau:G\rightarrow \{1, 2, \ldots, |G|\}$ and defining $h(g) = \tau(f(g))$. Define the influence of the $i^{th}$ coordinate of a function $h$ to be the expected variance of $h$ when the \ $i^{th}$ coordinate is sampled uniformly from $G$, with all other coordinates held fixed. The expectation is taken over uniformly random choices of the other coordinates. Analytically, in terms of the Fourier coefficients of $h$, the $i^{th}$ influence is equal to the following quantity:
	$$\inf_i(h):= \sum_{\alpha: \alpha_i\neq 0} |\hat{h}(\alpha)|^2.$$
	
	 Thus, the $i^{th}$ dictator function has $\inf_i(h) = \Omega_{G}(1)$ as a function of $G$ independent of $n$, and every other influences $0$. At the first attempt, it might make sense to define functions that are close to dictators as those with a coordinate that has a large influence. However, for a prime $p$ and $G\cong\mathbb Z_p$, note that there are linear functions $\ell_S = \sum_{i\in S}\beta_i x_i$ where $S\subseteq[n]$ such that $\ell_S$ has all variables $i\in S$ with influences $\Omega(1)$. We would like to isolate these functions with large $|S|$ from dictator functions. This motivates us to define a more refined notion of a low-degree influence of a variable $i$. For a function $h:G^n\rightarrow \mathbb{C}$, define its degree-$d$ influence of the $i^{th}$ coordinate as 
	
	$$ \inf_i^{\leq d}(h):= \sum_{\alpha: \alpha_i\neq 0 \wedge |\alpha|\leq d} |\hat{h}(\alpha)|^2,$$
	where $|\alpha|$ is the number of non-zero coordinates of $\alpha$. Thus, for the scalar encoding of the $i^{th}$ dictator function, its low-degree ($d=1$) influence of the coordinate $i$ is $\Omega(1)$ (and the rest of the influences are $0$). We consider a function to be far from any dictator function if all the low-degree influences, for some $d=O(1)$ that is independent of $n$, of the function are small, say at most $\eps$. 
	
	Although the above definition is the right definition for most reductions, for the hardness results of Max-$k$-{\LIN}, we work with an even weaker notion of the distance. We consider the following definition. We always assume the function is {\em folded}, i.e.,  it satisfies $f(a x) = a f(x)$ for all $a\in G$ and $x\in G^n$.\footnote{Such a condition can be enforced in the reduction because we allow group constants in each constraint.} A function is far from dictator functions if for every $\alpha$ such that $|\alpha| = O(1)$ which is independent of $n$, $|\hat{h}(\alpha)|^2 \leq\eps$. In other words, all {\em low-degree} Fourier coefficients of $h$  have small weights. Note that this notion isolates $\ell_S$ with large $|S|$ from dictators.
	
	A (non-adaptive) dictatorship test queries the function $f$ at a few locations and, based on the values it sees, decides whether the function is a dictator function or far from it. The choice of the accepting criterion for the test is tightly connected to the specific constraint satisfaction problem (CSP) for which we want to show an NP-hardness result. Furthermore, the gap between the test passing probability in the {\em completeness case} (when $f$ is a dictator function) and the {\em soundness case} (when $f$ is far from any dictator function) translates into the inapproximability factor of the NP-hardness of the CSP.\footnote{This is not totally correct as one has to overcome other important issues when such a test is used in the actual reduction, starting with the Label Cover instance.}
	
	\subsubsection{Abelian Groups}
	Let us look at a candidate dictatorship test where the predicate is a linear equation in $3$ variables over an abelian group $G\cong \mathbb{Z}_q$. Here, $0$ is the identity element of $G$.

\fbox{
	
	\parbox{250pt}{
\begin{itemize}
	\item Select $x, y \sim G^n$ uniformly at random.
	\item Set $z = x+y$.
	\item Check if $f(x) + f(y) = f(z) $.
\end{itemize}

}
}\\

	It is clear that if $f$ is an $i^{th}$ dictator, then the test passes with probability $1$. The non-trivial thing is to analyze the test passing probability if $f$ is far from any dictator.  It is easy to see that there are functions that are far from dictators, and still the test passes with probability $1$ on them. A family of such functions is the linear functions of the form $\ell_S = \sum_{i\in S}\beta_i x_i$ where $S\subseteq[n]$ and $\beta_i \in G\setminus \{0\}$, with large $|S|$. It is not hard to see that these functions pass the test with probability $1$. In fact, Blum, Luby, and Rubinfeld~\cite{BLR90} showed that this is a good test for the linear functions (instead of the dictator functions). 
	
	To prove the hardness of the approximation for Max-$3$-{\LIN} over abelian groups, one needs to design a test in which $\ell_S$, for large $|S|$, passes with only a small probability. To design such a test, H\aa stad~\cite{Hastad2001} introduced the so-called {\em noise} to each coordinate. The modified test is as follows:

	\fbox{
		
		\parbox{300pt}{
	\begin{itemize}
		\item Select $x, y \sim G^n$ uniformly at random.
		\item Set $z = x+y$.
		\item For each $i\in [n]$, resample $(x_i, y_i, z_i)$ from $G^3$ uniformly at random, with probability $\eps$. 
		\item Check if $f(x) + f(y) = f(z)$.
	\end{itemize}
}}\\
	
	This noise takes care of the earlier mentioned counterexamples, i.e., functions $\ell_S$ for large $S$ now pass the test with probability roughly $\frac{1}{|G|}$. In general, H\aa stad~\cite{Hastad2001} showed that if $f$ is far from dictator functions, then the test passes with probability at most $\frac{1}{|G|}+\delta$ for a small constant $\delta>0$. The proof of this statement uses Fourier analysis over abelian groups. Note that this bound is optimal as even a random function passes the test with probability $\frac{1}{|G|}$.
	
%	If we let $G = \mathbb{Z}_q$ then the probability that the test passes can be analyzed as follows:
%	We consider the function $g: G^n \rightarrow \C$ where $g(x) = \omega^{f(x)}$.
%	\begin{algin*}
%		\Pr[Test passes] = \E_{x, y, z}[]
	
	However, now the guarantee in the completeness case is no longer the same. We only get that dictator functions pass this test with probability $(1-\eps)+\frac{\eps}{|G|}$ (instead of $1$). This gap in the test passing probability is translated into the NP-hardness of $3$-\LIN\ over abelian groups and coincidentally, in this abelian case, the NP-hardness result is {\em optimal}. More precisely, given a system of linear equations over an abelian group $G$, where each equation involves $3$ variables, it is NP-hard to distinguish between the cases when there exists an assignment that satisfies at least $(1-\eps)$-fraction of the constraints vs. no assignment can satisfy more than $\frac{1}{|G|}+\delta$ fraction of the constraints.

		\subsubsection{Non-Abelian Groups} We now look into the non-abelian case. Since we would like to design a test which passes with probability $1$ (or $(1-\eps)$) in the completeness case, there is a natural generalization of the above-mentioned tests to a non-abelian group $G$. Here we denote the group operation by the symbol $\gp$ and the identity element of $G$ by $1_G$.
		
		We first describe the test with completeness $(1-\eps)$, which is similar to the test over the abelian group with noise that we described earlier.

		\fbox{
			
			\parbox{300pt}{
		\begin{itemize}
			\item Select $x, y \sim G^n$ uniformly at random.
			\item For each $i\in [n]$, set  $z_i  = y_i^{-1} \gp x_i^{-1}$.
			\item For each $i\in [n]$, resample $(x_i, y_i, z_i)$ from $G^3$ uniformly at random, with probability $\eps$. 
			\item Check if $f(x) \gp f(y) \gp f(z) = 1_G$.
		\end{itemize}
	}}\\
		
The analysis of this test is implicit in the work of Engebretsen, Holmerin, and Russell~\cite{EngebretsenHR}. Firstly, it is easy to see that the dictator functions pass this test with probability $(1-\eps)+\frac{\eps}{|G|}$. The soundness of this test is analyzed in~\cite{EngebretsenHR} where the authors show that in the soundness case, the test passes with probability at most $\frac{1}{|G|}+\delta$ for a small constant $\delta>0$. Their proof goes via  Fourier analysis over non-abelian groups. As in the abelian case, in this case also, it can be shown that the noise takes care of {\em high degree} Fourier terms.\footnote{Although we have not formally defined a degree of a Fourier coefficient of a function in this non-abelian setting, think of it as the number of non-trivial irreducible representations in $\alpha$ (See Proposition~\ref{prop:irr_Gn}). } This implies NP-hardness result of approximating $3$-\LIN\ instances over a non-abelian group, which is similar to the abelian case.

Although the proof of the soundness of the test in~\cite{EngebretsenHR} uses representation theory and Fourier analysis of functions on non-abelian groups, the proof now follows from the more general statement called the {\em invariance principle} of Mossel~\cite{Mossel10}. The distribution on the tuple $(x, y, z)$ is a product distribution $\mu^{\otimes n}$ where $\mu$ is a distribution on $(x_i, y_i, z_i)$ (note that for all $i$ it is the same distribution). Since we add noise to each coordinate with some non-zero probability, the distribution $\mu$ is {\em connected} and hence we can easily take care of high-degree functions in the analysis. Furthermore, the distribution $\mu$ is pairwise independent. These two conditions, along with the fact that $f$ has small low-degree influences, imply that the test passes with probability at most $\frac{1}{|G|}+\delta$. This statement is implicit in the result of Austrin and Mossel~\cite{AustrinM09}.

We now state the (obvious) dictatorship test with perfect completeness.
	
\fbox{
	
	\parbox{250pt}{
		\begin{itemize}
	\item Select $x, y \sim G^n$ uniformly at random.
	\item For each $i\in [n]$, set  $z_i  = y_i^{-1} \gp x_i^{-1}$.
	\item Check if $f(x) \gp f(y) \gp f(z) = 1_G$.
\end{itemize}
}}\\

Our main contribution is the soundness analysis of the above dictatorship test over non-abelian groups {\em without noise}. Note that we cannot use the invariance principle-based techniques in this case, as the distribution does not satisfy the condition of {\em connectedness}.

\paragraph{Proof overview.}

Our proof of the soundness analysis is inspired by the magic that was discovered by Gowers~\cite{Gowers08} to show that there are non-abelian groups where the size of any product free set is sublinear in $|G|$.\footnote{unlike the abelian case where one can always find, in this case, a 'sum-free' set of size $\Omega(|G|)$.} Gowers' trick applied only to quasi-random groups,\footnote{A group $G$ is called {\em quasi-random} if the smallest dimension of any non-trivial representation of $G$ is large.} as the goal was to obtain an $o(|G|)$ bound on product-free sets, whereas our reduction applies to every non-abelian group.

The trick is elegantly captured by the following inequality due to Babai, Nikolov, and Pyber~\cite{BNP08}. For any functions $f,g : G \rightarrow \C$  with at least one of $f, g$ having mean zero:
\begin{equation}
\label{eq:convimprovement}
\|f*g\|_{L^2(G)} \leq \frac{1}{\sqrt{D}}\|f\|_{L^2(G)} \|g\|_{L^2(G)},
\end{equation}
where $D$ is the smallest dimension of a non-trivial representation of $G$\footnote{see Definition~\ref{def:norm} and Definition~\ref{def:conv} for the definitions of $\|\cdot\|_{L^2(G)}$ and the convolution operator $*$, and Section~\ref{sec:reptheory} for representation theory;} while a trivial application of the Cauchy-Schwarz inequality gives an upper bound of $\|f\|_{L^2(G)} \|g\|_{L^2(G)}$. Thus, Equation~(\ref{eq:convimprovement}) has a multiplicative improvement of a factor $\frac{1}{\sqrt{D}}$ over a trivial upper bound.

%Engebretsen \etal~\cite{EngebretsenHR} showed the approximation resistance of Max-$k$-\LIN\ over {\em any} non-abelian group with imperfect completeness. Because of the ability to add {\em noise} to individual co-ordinates, it was easy to control the higher order terms that arise in the soundness analysis. They also use representation theory to analyze the soundness of the reduction. The result of \cite{EngebretsenHR} now follows from a more general result of Austrin-Mossel~\cite{AustrinM09}, which is based on invariance principle. When we go to the perfect completeness setting, the invariance principle technique completely breaks down. The main contribution of this work is to effectively bounding the higher order terms without using the {\em noise}. The overall approach of analyzing the soundness of the reduction is similar to the one in \cite{EngebretsenHR} which uses representation theory and Fourier analysis over non-abelian groups. However, our proof crucially uses the {\em quasi-randomness} of non-abelian groups as explained next.
	
Returning to analyzing the soundness of our test, its analysis boils down to analyzing the following expression.
\begin{align}
\left|\E[g_1(x) g_2(y)g_3(z)]\right|&= \left|\E[(g_1*g_2*g_3)(1_{G^n})]\right|\nonumber\\
& \leq  \sum_{\alpha \in \irr(G^n)} \dim(\alpha)\cdot \hsnorm{\hat{g_1}(\alpha)}\cdot \hsnorm{\hat{g_2}(\alpha)}\cdot \hsnorm{\hat{g_3}(\alpha)},\label{eq:fourier_sum}
\end{align}
where $g_i : G^n \to \mathbb{C}$ are bounded balanced functions obtained from $f$ (so $\|g_i\|_2 \le 1$). 
Here, $\widehat{g_i}(\alpha)$ denotes the Fourier coefficient of $g_i$ associated with the irreducible representation $\alpha$ of $G^n$, and $\hsnorm{\cdot}$ denotes the Hilbert--Schmidt norm of a matrix. 
The stated inequality follows by expanding $(g_1 * g_2 * g_3)$ in the Fourier basis and applying the triangle inequality.

Once we have this expression, similar to Equation~(\ref{eq:convimprovement}), it is easy to bound the terms with large $\dim(\alpha)$: By applying the Cauchy-Schwarz inequality and using Parseval's identity, we can show the following:
\begin{align}
\sum_{\dim(\alpha)\geq D} \dim(\alpha)\cdot \hsnorm{\hat{g_1}(\alpha)}\cdot \hsnorm{\hat{g_2}(\alpha)}\cdot \hsnorm{\hat{g_3}(\alpha)} \leq \frac{1}{\sqrt{D}} \|g_1\|_2 \cdot \|g_2\|_2\cdot \|g_3\|_2 \leq \frac{1}{\sqrt{D}}.\label{eq:intro_high_terms}
\end{align}
Thus, we can effectively bound the {\em higher dimensional} terms in the expression. Therefore, if the original expectation is $\delta$, then we essentially get 
$$\sum_{\dim(\alpha)\leq D} \dim(\alpha)\cdot \hsnorm{\hat{g_1}(\alpha)}\cdot \hsnorm{\hat{g_2}(\alpha)}\cdot \hsnorm{\hat{g_3}(\alpha)}= \Omega(\delta),$$
for $D = \Theta(1/\delta^2)$. Taking the maximum $ \hsnorm{\hat{g_1}(\alpha)}$ out in the summation, the remaining sum can be upper-bounded by $1$. Therefore, we get that there exists an $\alpha$ such that the dimension of $\alpha$ is at most $D$ and  $\hsnorm{\hat{g_1}(\alpha)}= \Omega(\delta)$.

Our analysis shows that if the test passes with probability greater than $\frac{1}{\left|\commutator{G}{G}\right|}+\delta$, then there exists an $\alpha$ such that $\dim(\alpha)\leq O_\delta(1)$ and  $\hsnorm{\hat{\tilde{f}}(\alpha)}\approx \delta$, for some function $\tilde{f}: G^n\rightarrow \mathbb{C}$ derived from $f$.  Note that this conclusion is different from what we had aimed for, i.e., concluding that there exists a {\em low-degree} Fourier coefficient with large magnitude. However, we show that such a bound is enough to carry out the actual soundness analysis of the reduction.

Since we do not introduce noise to each coordinate, there are functions with the Fourier mass concentrated on a large dimensional $\alpha$ that pass this test with probability $\frac{1}{\left|\commutator{G}{G}\right|}$. Thus, our analysis of the test is also optimal. 

%To see this, let $H  = G/\commutator{G}{G}$ where $\commutator{G}{G}$ is a commutator subgroup of $G$. Define a function $f(x_1, x_2, \ldots, x_n) = \sum x_i\cdot H$

Although the dictatorship test works, there are many complications that arise when we compose this test with the Label Cover instance. We briefly discuss three issues here and elaborate on the last one in this proof overview:
	\begin{enumerate}
	\item As observed before, in the soundness analysis, we conclude that there is a low-dimensional Fourier coefficient whose norm is large, if the test passes with non-trivial probability. However, there are terms with dimension $1$ but with a high degree, which are problematic for the final {\em decoding strategy} in our reduction. In ~\cite{EngebretsenHR}, these problematic terms were handled by adding noise, a technique similar to the Abelian case. In our case, we do not have the noise. However, in Lemma~\ref{claim:folded_zero_coeff}, we observe a stronger property of the folded functions.
    \footnote{A function $f$ is called folded if $f(c\gp x) = c\gp f(x)$ for all $x\in G^n$ and $c\in G$} 
    Namely, if $f$ is folded, then the function $\rho(f(x))_{ij}$, where $\rho$ is any irreducible representation of $G$ of dimension at least $2$, has all Fourier coefficients with dimension $1$ {\em zero}. Thus, we can focus on terms with dimensions at least $2$. 
	
	\item Our decoding strategy is different from that in~\cite{EngebretsenHR}. The decoding strategy in ~\cite{EngebretsenHR} is based on non-empty low-degree Fourier coefficients, which was similar to H\r{a}stad's decoding strategy. In our reduction, we slightly changed the decoding strategy --- it is based on Fourier coefficients whose dimension is at least $2$, but they can be of high degree. This condition is forced on us by the way we handle the higher-dimensional terms. Fortunately, the decoding strategy works without much trouble.
	
	\item Finally, due to the $d$-to-$1$ nature of the projection constraints in the \LC\ instance, the three queries tend to be more correlated to each other.  Thus, one has to take care of many potential scenarios in the actual reduction when the error term can be large. We handle this by starting with a variant of the {\LC} problem called the {\em layered} {\LC} whose \np-hardness was shown in~\cite{DGKR05}. The reason we had to work with a layered {\LC} is somewhat technical and here, we try to shed some light on it.
	\end{enumerate}

\paragraph{Controlling the terms with large dimension.}
In the dictatorship test, large representation dimension gives us a
useful saving almost automatically. After composing the test with
Label Cover, however, a large-dimensional representation can break
into small-dimensional pieces, and we must find another way to
control their contribution.

To see where this difficulty arises, fix a $d$-to-$1$ projection
$\pi:[dn]\to[n]$ and two folded tables
$f_1:G^{dn}\to G$ and $f_2:G^n\to G$.
Choose independent uniform inputs $\V a\in G^{dn}$ and $\V b\in G^n$,
set $c_i=b_{\pi(i)}^{-1}a_i^{-1}$, and check whether
\[
    f_1(\V a)\gp f_2(\V b)\gp f_1(\V c)=1_G.
\]
Consistent dictators pass this test with probability $1$: if
$f_1(\V x)=x_i$ and $f_2(\V y)=y_{\pi(i)}$, the three answers multiply
to the identity.

Each queried input is still marginally uniform on its own domain.
The difficulty comes from the lifted tuple $\V b\circ\pi$, which is
uniform on the subgroup
\[
    S_\pi:=\{\V x\circ\pi:\V x\in G^n\}\leq G^{dn}.
\]
The Fourier calculation therefore involves restricting representations
of $G^{dn}$ to $S_\pi$. An irreducible representation $\alpha$ of large
dimension may become reducible under this restriction and contain
irreducible constituents $\beta$ of very small dimension. Thus,
the dimension of $\alpha$ alone no longer gives the saving we used
in the dictatorship test.

Here is the expression that captures this obstruction. Let
$g_1:G^{dn}\to\C$ and $g_2:G^n\to\C$ be the $1$-bounded
matrix-coefficient functions derived from the two tables.
For $\alpha\in\irr(G^{dn})$ and $\beta\in\irr(G^n)$, let
$P_{\alpha,\pi,\beta}$ be the orthogonal projection onto the sum of
all copies of $\beta$ in the restriction of $\alpha$ to
$S_\pi\cong G^n$, with the zero projection if $\beta$ does not occur.
Multiplying $\hat g_1(\alpha)$ by this projection projects each of
its columns onto that subspace. After applying Cauchy--Schwarz,
the constituents with $\dim(\beta)>D$ contribute at most $1/D$
to the bound on the squared high-dimensional contribution.
The remaining quantity is
\begin{align}
\Delta
:=
\sum_{\substack{\beta\in\irr(G^n)\\ \dim(\beta)\leq D}}
    \hsnorm{\hat g_2(\beta)}^2
\sum_{\substack{\alpha\in\irr(G^{dn})\\ \dim(\alpha)\geq D}}
    \dim(\alpha)\,
    \hsnorm{P_{\alpha,\pi,\beta}\hat g_1(\alpha)}^2.
\label{eq:higher-intro}
\end{align}
The inner sum measures how much Fourier mass from large-dimensional
$\alpha$'s lands in copies of $\beta$. The outer factor measures
how strongly the other table uses that same $\beta$.
Consequently, $\Delta$ can be large when the two tables have
substantial Fourier mass on matching small-dimensional pieces,
even though the original $\alpha$'s have large dimension.

Our way of handling this possibility is to use many layers.
We start from the layered Label Cover construction of~\cite{DGKR05},
in which projections along a path compose consistently.
Sample a uniform path
$\mathrm p=(v_1,\ldots,v_\ell)$ through the layers, choose a uniform
pair $k<k'$, and apply the test above to the tables at $v_k$ and
$v_{k'}$. The structural properties of the construction ensure that
this samples a uniform constraint between the chosen layers.
Writing $\Delta_{\mathrm p,k,k'}$ for the corresponding quantity
in~\eqref{eq:higher-intro}, we show that a large average
$\E_{\mathrm p,k,k'}[\Delta_{\mathrm p,k,k'}]$ would give a
labeling satisfying a nontrivial fraction of the Label Cover
constraints.

The geometric idea is simple: a vector has only a bounded amount
of squared norm to distribute among mutually orthogonal subspaces.
Fix a path and an initial layer $k$. As we move to later layers,
the corresponding subgroups of $G^{R_k}$ form a decreasing chain.
We can therefore choose the individual irreducible-copy subspaces
successively, so that each block at a later layer lies inside a
block at every earlier layer. This lets us follow the same
Fourier directions through all the later layers.

Suppose the average contribution is at least $\eta>0$.
Representations $\beta$ for which
$\hsnorm{\hat g_2(\beta)}^2$ is much smaller than $\eta$ account
for only a small part of this contribution.
Call the remaining representations \emph{heavy}.
By Parseval, there are only $O_\eta(1)$ heavy representations at
each vertex. If the subspaces corresponding to heavy representations
at different later layers were always orthogonal, their contributions
would have to share the bounded Fourier mass of the initial table.
With sufficiently many layers, their average would then be small.
A quantitative version of this observation shows that, for some
fixed pair of later layers, heavy components overlap on a
nontrivial fraction of paths.

Finally, this overlap gives consistent labels. For a representation
$\beta=\bigotimes_j\rho_j$, consider the coordinates $j$ for which
$\dim(\rho_j)\geq2$. Since $\dim(\beta)\leq D$, there are at most
$\log_2 D$ such coordinates. Folding ensures that this set is
nonempty for every heavy representation used in the argument.
Moreover, when a later block lies inside an earlier block,
the later representation is a constituent of the restriction of
the earlier one. Every coordinate of dimension at least two in
the later representation must therefore be the image, under the
Label Cover projection, of such a coordinate in the earlier
representation.

Taking the union of these coordinate sets over the heavy
representations gives each vertex a short list of candidate labels.
The overlapping components guarantee that these lists contain a
consistent pair on a nontrivial fraction of constraints.
Choosing labels independently and uniformly from the lists
therefore satisfies an $\Omega_{\eta,D}(1)$ fraction of the
constraints between two layers, contradicting the soundness of
the layered Label Cover instance.

\section{Preliminaries}

\subsection{Label Cover}

We start by defining the \LC\ problem.

\begin{definition}[\LC]
	
	\label{def:label-cover} An instance $\calH=(\calU,\calV,E,[L],
	[R],\{\pi_e\}_{e\in E})$ of the {\LC} constraint satisfaction
	problem consists of a bi-regular bipartite graph $(\calU,\calV,E)$,
	two alphabets $[L]$ and $[R]$ and a surjective projection map $\pi_e : [R] \rightarrow
	[L]$ for every edge $e\in E$. 
	Given a labeling $\ell_{\calU} : \calU \rightarrow [L], \ell_{\calV}:\calV \rightarrow
	[R]$, an edge $e = (u,v)$ is said to be satisfied by $\ell$ if
	$\pi_e(\ell_{\calV}(v)) = \ell_{\calU}(u)$. 
	
	$\calH$ is said to be \emph{satisfiable} if there exists a 
	labeling that satisfies all the
	edges. $\calH$ is said to be \emph{at most $\delta$-satisfiable} if every
	labeling satisfies at most a $\delta$ fraction of the
	edges. 
\end{definition}

The hardness of \LC\ stated below follows from
the PCP Theorem \cite{AroraS1998,AroraLMSS1998, FGLSS96} and Raz's Parallel
Repetition Theorem \cite{Raz1998}. 
\begin{theorem}[Hardness of \LC]
	
	\label{thm:lc-hard} For every $r \in \N$, there is a
	deterministic $n^{O(r)}$-time reduction from a \TSAT\ instance
	of size $n$ to an instance $\calH=(\calU,\calV,E,[L],[R], \{\pi_e\}_{e
		\in E})$ of \LC\ with the following properties:
	
	\begin{enumerate}
		
		\item $|\calU|,|\calV| \leq n^{O(r)}$; $L,R \leq 2^{O(r)}$; $\calH$ is bi-regular with degrees bounded by $2^{O(r)}$.

		\item There is a constant $s_0 \in (0,1)$ such that,
		
		\begin{itemize}
			
			\item YES Case: If the \TSAT\ instance is
			satisfiable, then $\calH$ is satisfiable.
			
			\item NO Case: If the \TSAT\ instance is
			unsatisfiable, then $\calH$ is
			at most $2^{-s_0r}$-satisfiable.
		\end{itemize}
	\end{enumerate}
\end{theorem}

For our hardness result, we need the following variant of the \LC\ problem.

\begin{definition}[Layered Label Cover]
\label{def:layered-LC}
An $\ell$-layered Label Cover instance, given by $\calH=(\calV = \{V_1,\ldots, V_\ell\},\{\Pi_{i,j}\}_{1\leq i<j\leq \ell},\{[R_i]\}_{i\in [\ell]})$ consists of $\ell$ sets of vertices $\calV = \{V_1,\ldots, V_\ell\}$. The label set of vertices in layer $i$ is denoted by $[R_i]$. Every pair of layers $1\leq i < j \leq \ell$ has a set of constraints $\Pi_{i,j}$ between the vertices in $V_i$ and $V_j$. The constraint between $v\in V_i$ and $u\in V_j$ (if one exists in $\Pi_{i,j}$) is denoted by $\pi_{vu}:[R_i]\rightarrow [R_j]$. Moreover, every constraint between a pair of vertices is a projection constraint  -- for every assignment $k \in [R_i]$ to $v\in V_i$ there is a unique assignment to $u\in V_j$ that satisfies the constraint $\pi_{vu}$.
\end{definition}

%In a Layered Label Cover instance as defined above, for any constraint $\pi_{v \rightarrow u} \in \Pi_{i,j}$, we view it as a function $\phi_{x \rightarrow y}: [R_i] \rightarrow [R_j]$ defined such that for any $k \in [R_i]$, $(k, \phi_{x \rightarrow y}(k))$ satisfies the constraint $\phi_{x \rightarrow y}$. Thus, where there is no ambiguity, we will use $\phi_{x \rightarrow y}$ to denote both the constraint, as well as the function. Moreover, for brevity, we say $x \sim y$, or ``$x$ is a neighbour of $y$'' if $\phi_{x \rightarrow y} \in \Phi_{i,j}$.

We have the following $\np$-hardness result from \cite{DGKR05}, \cite{DRS05}, which we use as a starting point in proving  Theorem~\ref{thm:mainnphardness}. The additional structural properties mentioned in the theorem statement can be derived easily from their construction. For completeness, we include the proof of this in the appendix.

\begin{theorem}[\cite{DGKR05}, \cite{DRS05}]
\label{theorem:multiLC}
For any constant parameters $\ell \geq 2, r \in \mathbb{N}$ the following problem is $\np$-hard.
Given an $\ell$-layered Label Cover instance $\calH=(\calV= \{V_1,\ldots, V_\ell\},\{\Pi_{i,j}\}_{1\leq i<j\leq \ell},\{[R_i]\}_{i\in [\ell]})$ where all variable ranges $[R_i]$ are of size $R_i \leq 2^{O(\ell r)}$, distinguish between the following two cases:
\begin{description}
\item[Completeness.] There is an assignment satisfying all the constraints of the Label Cover instance. In this case, we say that $\calH$ is fully satisfiable.
\item[Soundness.] For every $1 \le i < j \le \ell$, no assignment satisfies more than a $2^{-r}$ fraction of the set of constraints $\Pi_{i,j}$ between layers $i$ and $j$. In this case, we say that $\calH$ is at most $2^{-r}$-satisfiable.
\end{description}
Furthermore, the above hardness result also holds for instances with the following additional structure.
\begin{enumerate}
 \item {\em Transitivity of the constraints.} The set of constraints $\Pi_{i, i+1}$ for $1\leq i\leq \ell-1$ define all the other constraints by transitivity. In other words, fix any three layers $1\leq i<j<k\leq \ell$. For any two constraints $\pi_{vu} \in \Pi_{i,j}$ and $\pi_{uw}\in \Pi_{j, k}$, there is a constraint $\pi_{vw}\in \Pi_{i,k}$ such that $\pi_{vw} (a) = \pi_{uw}(\pi_{vu}(a))$ for all $a\in [R_i]$.
 \item {\em Regularity of the constraints.} The constraint graph $G_{i,j}$ is biregular for every $1\leq i<j\leq \ell$. Moreover, the following distributions on constraints are identical:
\begin{enumerate}
\item[(a)] Choose a layer pair $i<j$ uniformly, then a constraint uniformly
  from $\Pi_{i,j}$.
\item[(b)] Choose a full consecutive-layer chain
  $(x_1,\ldots,x_\ell)$ uniformly, independently choose a layer pair
  $i<j$ uniformly, and output its induced constraint $(x_i,x_j)$.
\end{enumerate} 
    \item {\em Regular projection constraints.} There exists a constant $d$ such that $R_{i-1} = d R_{i}$ for all $2\leq i\leq \ell$ and each constraint $\pi_{vu}\in \Pi_{i,j}$ is a $d^{(j-i)}$-to-$1$ projection constraint.

\end{enumerate}

\end{theorem}

%In the appendix, we give a construction of the layered \LC\ from~\cite{DGKR05} and observe the above-mentioned structure on the hard instances.  \abnote{Do it.}

\subsection{Fourier analysis}

	 In this section, we give a brief overview of the representation theory of non-abelian groups and Fourier analysis over non-abelian groups. For a more comprehensive understanding, we refer the reader to the book by Terras~\cite{Terras99}. We state many propositions in the following subsection, and the proofs of these propositions can be found in the same book \cite{Terras99}.
	 
\subsubsection{Representation Theory}
\label{sec:reptheory}
In this paper, we only consider non-abelian groups which are {\em finite}. Let $G = (G,\gp)$ be a finite non-abelian group. The identity element of a group is denoted by $1_G$.
	
	\begin{definition}
		A representation $(V, \rho)$ of $G$ is a vector space $V$ together with a group homomorphism $\rho: G \rightarrow \GL(V)$ from $G$ to the group $\GL(V)$ of invertible $\C$-linear transformations from $V$ to $V$. The dimension of the vector space $V$ is denoted by $\dim(\rho)$.
	\end{definition} 

For convenience, we just use the letter $\rho$ to denote a representation of $G$ and use $\rho_V$ to denote the underlying vector space. We view a representation $\rho(\cdot)$ as its corresponding matrix of the linear transformation. Thus $\rho(\cdot)_{ij}$ is used to denote the $(i,j)^{th}$ entry of that matrix. We always work with representations that are unitary. There is one representation which is obvious -- just map everything to $1\in \C$.  This representation is called the {\em trivial representation}, which has dimension $1$. We will denote the trivial representation by $\trivrep$. 

\begin{definition}
	Let $\rho$ and $\tau$ be representations of $G$. An isomorphism from $\rho_V$ to
	$\tau_V$ is an invertible linear transformation $\phi : \rho_V \rightarrow \tau_V$ such that 
	$$ \phi \circ \rho(g) = \tau(g)\circ \phi,$$
	for all $g\in G$. We say that $\rho_V$ and $\tau_V$ are isomorphic and write $\rho_V\cong \tau_V$ if there exists an isomorphism from $\rho_V$ to $\tau_V$.
\end{definition}

\begin{definition}
	Let $\rho$ be a representation of $G$. A vector subspace $W\subset \rho_V$ is $G$-invariant if $\rho(g)w \in W$ for all $g\in G$ and $w\in W$.
\end{definition}

If a representation $(V, \rho)$ has a $G$-invariant subspace $W$ other than $\{0\}$ and $V$ itself, then the action on $W$ itself is a representation of $G$. This leads to the following important definition of irreducible representations.

\begin{definition}
A representation $\rho$ of $G$ is irreducible if $\rho_V\neq \{0\}$ and $\rho_V$ has no $G$-invariant subspaces other than $\{0\}$ and $\rho_V$.
\end{definition}

We will denote the set of all irreducible representations of $G$ up to isomorphism by $\irr(G)$.

\begin{fact}
	\label{fact:subgroup_rep}
	Let $G$ be a group and $H$ be any subgroup of $G$. If $\rho\in \irr(G)$  then $\rho$ restricted to $H$ is also a (not necessarily irreducible) representation of $H$.
\end{fact}

\begin{definition} The tensor product of two representations $\rho$ and $\tau$ of a group $G$  is
the representation $\rho \otimes \tau$ on $\rho_V \otimes \tau_V$ defined by the condition
$$(\rho\otimes \tau)(g)(v\otimes w) = \rho(g)(v)\otimes \tau(g)(w),$$
and extended to all vectors in $\rho_V \otimes \tau_V$  by linearity.
\end{definition}

\begin{definition}
	The direct sum of two representations $\rho$ and $\tau$ is the space $\rho_{V} \oplus \tau_{V}$	with the block-diagonal action $\rho\oplus \tau$ of $G$.
\end{definition} 

If the representation is not irreducible, then by an appropriate change of basis $\rho$ can be converted into a block diagonal matrix with blocks corresponding to the invariant subspaces. Thus, any representation can be completely decomposed into a direct sum of irreducible representations of $G$, by applying an appropriate unitary transformation. Note that  this decomposition is {\em unique} (the invariant subspaces and change of basis are not unique; only the irreducible isomorphism classes and their multiplicities are unique).  We use the following notation to denote the decomposition of a reducible representation: If $\rho$ is a reducible representation of $G$ then $\rho \cong \oplus_{i} n_i \rho_i$, where for each $i$ we have {\em distinct} $\rho_i\in \irr(G)$ and $n_i$ denotes the multiplicity of $\rho_i$ in the decomposition. It will be convenient to think of this representation as a block diagonal matrix with $\rho_i$ as the blocks along the diagonal with multiplicity $n_i$.

We will use the following proposition which follows from an appropriate change of basis of the underlying space.

\begin{proposition}
\label{prop:refinement}
Let $A\le B\le G$, and let $\rho:G\to U(V)$ be a finite-dimensional
unitary representation. There are orthogonal decompositions
\[
V=\bigoplus_{i=1}^{r}W_i,
\qquad
W_i=\bigoplus_{j=1}^{s_i}V_{i,j},
\]
such that each $W_i$ is $B$-invariant and the representation of $B$
on $W_i$ is irreducible, while each $V_{i,j}$ is $A$-invariant and
the representation of $A$ on $V_{i,j}$ is irreducible.

An orthonormal basis adapted to these decompositions gives a single
unitary change of basis $U$ for which
\[
U\rho(b)U^*=\bigoplus_{i=1}^{r}M_i(b)
\quad (b\in B),
\qquad
U\rho(a)U^*=\bigoplus_{i=1}^{r}\bigoplus_{j=1}^{s_i}N_{i,j}(a)
\quad (a\in A).
\]
Here each block $M_i$ is the action on one irreducible $B$-summand
$W_i$, and each block $N_{i,j}$ is the action on one irreducible
$A$-summand $V_{i,j}$. Equivalent summands remain separate blocks;
their indices count individual copies, not isomorphism classes.
\end{proposition}

% % --- Proposition: block form for B and refinement for A ---
% \begin{proposition}
% \label{prop:refinement}
% Let $A\le B\le G$ be subgroups and let $\rho:G\to U(V)$ be a finite-dimensional
% unitary representation. Then there exist orthogonal decompositions
% \[
% V \;=\; \bigoplus_{i=1}^r W_i,
% \qquad
% W_i \;=\; \bigoplus_{j=1}^{s_i} V_{i,j},
% \]
% such that each $W_i$ is $B$-irreducible and each $V_{i,j}$ is $A$-irreducible.
% Consequently, there exists an orthonormal basis of $V$ adapted to the
% decomposition $V=\bigoplus_{i=1}^r\bigoplus_{j=1}^{s_i} V_{i,j}$, and letting
% $U$ be the associated unitary change-of-basis matrix, the conjugated
% representation
% \[
% \rho'(g)\;:=\;U\,\rho(g)\,U^*
% \]
% satisfies:
% \begin{enumerate}
% \item (\emph{Coarse block diagonalization on $B$}) For every $b\in B$,
% \[
% \rho'(b)
% =
% \begin{pmatrix}
% M_1(b) &        & 0\\
%        & \ddots & \\
% 0      &        & M_r(b)
% \end{pmatrix},
% \qquad\text{where } M_i(b)\in U(W_i).
% \]
% \item (\emph{Refined block diagonalization on $A$}) For every $a\in A$,
% \[
% \rho'(a)
% =
% \begin{pmatrix}
% \diag\!\big(N_{1,1}(a),\dots,N_{1,s_1}(a)\big) &        & 0\\
%                                               & \ddots & \\
% 0                                             &        &
% \diag\!\big(N_{r,1}(a),\dots,N_{r,s_r}(a)\big)
% \end{pmatrix},
% \]
% where $N_{i,j}(a)\in U(V_{i,j})$ for all $i,j$.
% \end{enumerate}
% In particular, the $A$-block structure is a refinement of the $B$-block
% structure in the sense that each $W_i=\bigoplus_{j=1}^{s_i} V_{i,j}$.
% \end{proposition}

The following proposition gives Schur orthogonality for matrix coefficients
with respect to a symmetric bilinear form. Matrix coefficients belonging to
non-isomorphic irreducible representations are orthogonal. For a fixed
irreducible representation, the pairing is zero unless the corresponding
indices are reversed, in which case it equals the reciprocal of the
dimension of the representation.
\begin{proposition}
	\label{prop:orthmatrixentries}
	If $\rho$ and $\tau$ are two non-isomorphic irreducible representations of $G$ then for any $i,j,k,\ell$ we have
	\begin{equation}
	\label{eq:diffG_Product}
	\langle  (\rho)_{ij}\mid ( \tau)_{k\ell}\rangle_G = 0,
	\end{equation}
	where $\langle  f_1\mid  f_2\rangle_G := \frac{1}{|G|}\sum_{g\in G} f_1(g) f_2(g^{-1})$ (called a ``symmetric bilinear form''). Also,
	\begin{equation}
	\label{eq:sameG_Product}
	\langle  (\rho)_{ij}\mid ( \rho)_{k\ell}\rangle_G = \frac{\delta_{i\ell}\delta_{jk}}{\dim(\rho)},
	\end{equation}
	where $\delta_{ij}$ is the Kronecker delta function which is $1$ if $i=j$ and $0$ otherwise. 
\end{proposition}

\subsubsection{Fourier analysis on non-abelian groups}
In this paper, we will be interested in studying $L^2(G)$, the space of functions from
a finite group $G$ to the complex numbers $\C$. 

\begin{definition}
	\label{def:norm}
	Define the inner product $\langle \cdot, \cdot\rangle_{L^2(G)}$ on $L^2(G)$ by 
	$$ \langle f, g \rangle_{L^2(G)} = \E_{x\in G} [f(x) \comp{g(x)}].$$
\end{definition}

We can define a character for every representation of a group.
\begin{definition} 
	The character of a representation $\rho$ is the function $\chi_\rho : G \rightarrow \C$ defined by 
	$\chi_\rho(g) = \trace(\rho(g))$.
\end{definition}

The following proposition shows that the characters corresponding to the irreducible representations of a group are orthogonal to each other.

\begin{proposition}[Orthogonality of characters]
	\label{prop:orthochar}
	For $\rho, \tau\in \irr(G)$, we have
	\begin{equation*}
	\frac{1}{|G|}\sum_{g\in G}\chi_{\rho}(g) \overline{\chi_{\tau}(g)}= \begin{cases}
	1\quad if~\rho_V \cong \tau_V,\\
	0 \quad otherwise.
	\end{cases}
	\end{equation*}
	
\end{proposition}

We use Proposition~\ref{prop:orthmatrixentries} many times in the proof. For convenience, we note an important identity that follows from  Proposition~\ref{prop:orthmatrixentries} (by setting $\tau$ to be the trivial representation $\trivrep$).
\begin{proposition}
	\label{prop:summapzero}
If $\rho\in \irr(G) \hspace{-3pt}\setminus \hspace{-3pt}\trivrep$, then $\sum_{g\in G} \rho(g) = 0$. 
\end{proposition}

We have the following proposition. It also shows that the maximum dimension of any irreducible representation of $G$ is at most $\sqrt{|G|}$.
\begin{proposition}
	\label{prop:sumdimsquare}
			$$\sum\limits_{\rho \in \irr(G)} \dim(\rho)\chi_\rho(g) = \begin{cases}
|G| \quad if ~ g= 1_G,\\
	0 ~~\quad otherwise.
	\end{cases} $$
This implies the following:
	$$\sum_{\rho\in \irr(G)} \dim(\rho)^2 = |G|.$$
\end{proposition}

\begin{definition}
	\label{def:conv}
For two functions $f, g\in L^2(G)$ their {\em convolution} $f\conv g \in L^2(G)$ is defined as 
$$(f\conv g)(x) := \E_{y\in G}[f(y)g(y^{-1}x)].$$
\end{definition}

	For an abelian group, any function $f: G \rightarrow \C$ can be written as a linear combination of characters, i.e., the characters span the whole space $L^2(G)$. However, for non-abelian groups, characters of irreducible representations form an orthonormal basis only for the set of {\em class functions} -- maps which are constant on {\em conjugacy classes}. A conjugacy class in $G$ is a nonempty subset $H$ of $G$ such that the following two conditions hold: 	Given any $x,y \in H$, there exists $g \in G$ such that $gxg^{-1} = y$, and if $x \in H$ and $g \in G$ then $gxg^{-1} \in H$. Since this is an equivalence relation, any group is a collection of disjoint conjugacy classes.

As in the abelian case, we can understand operations like inner product, convolution, etc., using the Fourier transform, which is defined as follows. 

\begin{definition}
For a function $f\in L^2(G)$, define the Fourier transform of $f$ to be the element $\hat{f} \in \prod_{\rho\in \irr(G)} \End{\rho_V}$ given  by  
$$\hat{f}(\rho) = \E_{x\in G}[ f(x) \rho(x)] \in \End{\rho_V}.$$
Here $\rho: G \to \mathrm{GL}(\rho_V)$ is an irreducible representation acting on the complex vector space $\rho_V$, and $\End(\rho_V)$ denotes the space of $\mathbb{C}$-linear endomorphisms of $\rho_V$ (equivalently, the algebra of all linear operators $\rho_V \to \rho_V$).

\end{definition}

\begin{definition}
Let $V$ be a finite-dimensional complex inner product space. Define an inner product  $\langle \cdot , \cdot\rangle_{\End{V}}$ on $\End{V}$ by
$$ \langle A, B\rangle_{\End{V}} = \trace(A\ctranspose{B}).$$
\end{definition}

We can now state the Fourier inversion theorem.
\begin{proposition}[Fourier inversion theorem]
	For $f\in L^2(G)$ we have
	$$f(x) = \sum_{\rho\in \irr(G)} \dim(\rho)\cdot \langle \hat{f}(\rho) , \rho(x)\rangle_{\End{\rho_V}}.$$
\end{proposition}

We have the following simple identities (See \cite{Terras99} for the proofs).
\begin{proposition}[Plancherel's identity] For  $f,g\in L^2(G)$
	\label{prop:plancherels}
	$$ \langle f, g \rangle_{L^2(G)} = \sum_{\rho\in \irr(G)} \dim(\rho)\cdot  \langle \hat{f}(\rho) , \hat{g}(\rho)\rangle_{\End{\rho_V}}.$$
\end{proposition}

\begin{proposition}[Parseval's identity]  For $f\in L^2(G)$
	\label{prop:parsevals}
	$$ \E_{x\in G}[|f(x)|^2] = \sum_{\rho \in \irr(G)} \dim(\rho)\cdot \hsnorm{\hat{f}(\rho)}^2,$$
	where $A\in\End{V}$, $\hsnorm{A}:= \sqrt{\langle A, A\rangle_{\End{V}}} = \sqrt{\trace(A\ctranspose{A})}= \sqrt{\sum_{ij} |A_{ij}|^2}$.
\end{proposition}

Note that the norm $\hsnorm{\cdot}$ satisfies submultiplicativity.
\begin{claim}
	\label{claim:normteq}
	$\hsnorm{AB}\leq \hsnorm{A}\cdot \hsnorm{B}.$
\end{claim}
\begin{proof}
$
	\hsnorm{AB}^2 = \sum_{ij} |(AB)_{ij}|^2 \leq \sum_{ij} \left(\sum_{k} |A_{ik} B_{kj}|\right)^2 
$.
	Using the Cauchy-Schwarz inequality on the inner sum,
	\begin{align*}
\hsnorm{AB}^2 	\leq  \sum_{ij} \left(\sum_{k} |A_{ik}|^2 \right)\left(\sum_{\ell} |B_{\ell j}|^2\right) = \sum_{ijk\ell} |A_{ik}|^2 |B_{\ell j}|^2 =  \left(\sum_{ik} |A_{ik}|^2 \right) \left(\sum_{\ell j} |B_{\ell j}|^2 \right) = \hsnorm{A}^2 \cdot \hsnorm{B}^2.
	\end{align*}
\end{proof}

\begin{claim}
	\label{claim:Uni_norm}
	Let $A$ be any matrix and $U$ be any unitary matrix, then $\hsnorm{UA} = \hsnorm{AU} = \hsnorm{A}$.
\end{claim}
\begin{proof}
	Using the identity that $\trace(AB) = \trace(BA)$, we have
$$\|UA\|_{\mathrm{HS}}^2=\trace(UAA^\star U^\star)=\trace(AA^\star)= \hsnorm{A}^2.$$
Similarly,
$$\|AU\|_{\mathrm{HS}}^2=\trace(AUU^\star A^\star)=\trace(AA^\star)= \hsnorm{A}^2.$$
\end{proof}
	
\begin{proposition}[Convolution theorem]
\label{prop:convolution}
	For $f, g\in L^2(G)$ we have 
	$$\hat{f\conv g}(\rho) = \hat{f}(\rho) \hat{g}(\rho).$$
\end{proposition}

The following claim shows that the character functions always come in 'pairs' with respect to the complex conjugation.
\begin{claim}
	\label{claim:dim1conj}
	Let $G$ be any non-abelian group. For every $\rho \in \irr(G)$, such that $\dim(\rho)=1$, there exists $\widetilde{\rho}\in \irr(G)$ with $\dim(\widetilde{\rho}) = 1$ such that 
	$$ \chi_\rho(g) = \overline{\chi_{\widetilde{\rho}}(g)}, \quad\quad \forall g\in G.$$
\end{claim}
\begin{proof}
	We claim that the set of characters corresponding to dimension $1$ irreducible representations of $G$ forms a group under point-wise multiplication. This will be enough to show the claim.

	Let $G' = G/\commutator{G}{G}$ be the abelian quotient group. Assume $\rho$ is a degree $1$ representation of $G$. Then it satisfies $\rho(a)\rho(b) = \rho(ab)$ for all $a, b\in G$. Define a map $\Gamma_\rho : G' \rightarrow \C$ as $\Gamma_\rho(g') = \rho(g)$ where $g' = g\commutator{G}{G}$. This is a well-defined map as
	$$\rho(aba^{-1}b^{-1}) = \rho(a)\rho(b) \rho(a^{-1})\rho(b^{-1}) = \rho(a)\rho(a^{-1})\rho(b)\rho(b^{-1}) = 1.$$ 
	Thus, the map $\rho$ is constant on every coset of $\commutator{G}{G}$ and hence $\Gamma_\rho$ is well defined. The set of all $\{\Gamma_\rho \mid \rho \in \irr(G), \dim(\rho) = 1\}$ is the set of all the multiplicative characters of the abelian group $G'$ and hence forms a group under coordinate-wise multiplication. There is a one-to-one correspondence between the coordinate-wise multiplicative action of $\Gamma_\rho$'s and $\rho$'s. Thus, $\{\chi_\rho \mid \rho \in \irr(G), \dim(\rho) = 1\}$ forms a group under point-wise multiplication.
\end{proof}

    \subsection{Functions on \texorpdfstring{$G^n$}{Gn}}

For any non-abelian group $G$ and $n\geq 1$, we have a group $G^n$ where the group operation is defined coordinate-wise.
The irreducible representations of $G^n$ are precisely those representations obtained by taking tensor products of $n$ irreducible representations of $G$.
\begin{proposition}[\cite{Terras99}]
	\label{prop:irr_Gn}
The set of irreducible representations of $G^n$ is given by 
	$$\irr(G^n) = \{ \alpha \mid  \alpha = \otimes_{i\in [n]} \rho_i \mbox{ where } \rho_i \in \irr(G)\}.$$
\end{proposition}
We denote $\alpha$ by the corresponding tuple $(\rho_1, \rho_2, \ldots, \rho_n)$. We define the weight of a representation $\alpha = (\rho_1, \rho_2, \ldots, \rho_n)$ (denoted by $|\alpha|$) to be the number of non-trivial representations in $(\rho_1, \rho_2, \ldots, \rho_n)$.

We will be working with functions $f: G^n \rightarrow G$ which are {\em folded}. $f$ is said to be folded if $f(c\V x) = cf(\V x)$ for all $c\in G$ and $\V x \in G^n$. The following claim shows that for all functions $g(\V x):=\rho(f(\V x))_{ij}$ where $\dim(\rho)\geq 2$ and $1\leq i,j\leq \dim(\rho)$,  all the Fourier coefficients corresponding to representations of dimension $1$ are zero if $f$ is folded.

\begin{lemma}
	\label{claim:folded_zero_coeff}
Let $f:G^n \rightarrow G$ be any folded function and $g(\V x):=\rho(f(\V x))_{ij}$ where $\rho \in \irr(G), \dim(\rho)\geq 2$ and $1\leq i,j\leq \dim(\rho)$. If $\alpha$ is {\em any} representation of $G^n$ such that $\dim(\alpha) = 1$, then $\hat{g}(\alpha) = 0$.
\end{lemma}
\begin{proof}
	Recall, for any $\V x\in G^n$, $\alpha(\V x)$ is a scalar as $\dim(\alpha) = 1$. $f$ is folded which means that $f(c \V x) = c f(\V x)$ for all $c\in G$ and $\V x \in G^n$. For $c\in G$, let $\V c \in G^n$ be such that every coordinate of $\V c$ is $c$. Since $\rho(.)$ has dimension at least $2$, in the following analysis, we use $\matrixb{\rho(.)}$ to denote the matrix of linear transformation for clarity.
	\begin{align*}
	\hat{g}(\alpha) = \E_{\V x} [ g(\V x)\alpha(\V x)]   &= \E_{\V x} [ \matrixb{\rho(f(\V x))}_{ij}\cdot \alpha(\V x)] \\
	& = \frac{1}{|G|} \E_{\V x} \left[ \sum_{c\in G} \matrixb{ \rho(f(c \V x))}_{ij} \cdot \alpha(c \V x)\right]\\
	& = \frac{1}{|G|} \E_{\V x} \left[ \sum_{c\in G}  \matrixb{\rho(cf(\V x))}_{ij}\cdot \alpha(\V c) \alpha(\V x)\right]\\
	& = \frac{1}{|G|} \E_{\V x} \left[ \sum_{c\in G}  \left( \alpha(\V c) \matrixb{\rho(c)}\cdot \matrixb{\rho(f(\V x))}\right)_{ij} \alpha(\V x)\right]\\
	& = \frac{1}{|G|} \E_{\V x} \left[  \left(\left(\sum_{c\in G} \alpha(\V c) \matrixb{\rho(c)}\right)\cdot \matrixb{\rho(f(\V x))}\right )_{ij} \alpha(\V x)\right].
	\end{align*}
	Now, for $\alpha\in \irr(G^n)$, let  $\widetilde{\alpha}\in \irr(G^n)$ be the dimension $1$ representation satisfying the condition in Claim~\ref{claim:dim1conj}. We have:
	\begin{align*}
	\sum_{c\in G} \alpha(\V c) \matrixb{\rho(c)}	& = \sum_{c\in G} {\overline{{\widetilde{\alpha}}(\V c)}}\cdot  \matrixb{\rho(c)}\\
	& = \sum_{c\in G} {{\widetilde{\alpha}}(\V c^{-1})}\cdot  \matrixb{\rho(c)} \\
	& = \sum_{c\in G} { ( \otimes_{i=1}^n {\widetilde{\alpha}_i}(c^{-1}))}\cdot  \matrixb{\rho(c)} \\
	& = \sum_{c\in G} { \tau(c^{-1})}\cdot  \matrixb{\rho(c)} \tag*{\text{$\dim(\tau)=1$}}\\
	& = 0,  \tag*{(Using Proposition~\ref{prop:orthmatrixentries})}
	\end{align*}
	where in the second last step, we used the fact that the product of dimension $1$ representations ($\otimes_{i=1}^n {\widetilde{\alpha}_i}$) of $G$ is itself a dimension $1$ representation ($\tau$) of $G$. Therefore, $\hat{g}(\alpha) = 0$.

\end{proof}

Fix any surjective projection map $\pi: [R]\rightarrow [L]$ for some $R\geq L$. Consider the following subgroup of $G^R$ given by the elements 
$$ \{ (x\circ \pi) \in G^R \mid x\in  G^L\},$$
where $(x\circ \pi)_i = x_{\pi(i)}$. Let us denote this group by $\pi(G^R)$. Note that this group is isomorphic to $G^L$. Thus, any representation $\alpha \in \irr(G^R)$ (which is a representation of $G^L$ using Fact~\ref{fact:subgroup_rep}) can be decomposed into irreducible representations of a subgroup of $G^R$ isomorphic to $G^L$. 

\subsection{Notations}

Whenever possible, we use the notation $\alpha, \beta$ to denote the representations of the group $G^n$ and $\rho, \tau$ for the group $G$. Also, we use bold letters $\V x, \V c$ to denote the elements of $G^n$.

For a representation $\alpha\in \irr(G^n)$ where $\alpha = \otimes_{i=1}^n \rho_i$, we use the notation $\dimgeqi{\alpha}{k}$ to denote the number of $i\in [n]$ such that $\dim(\rho_i)\geq k$.

\section{Warm-up: Dictatorship Test}

In this section, we analyze the dictatorship test, which involves checking some linear equation over a non-abelian group. The analysis will highlight a few important differences between our test and the linearity test over abelian groups.

Fix a non-abelian group $G$. Let $f: G^n \rightarrow G$ be a function. A function is called a dictator function if it is of the form $f(\V x) = x_i$ for some $i\in [n]$. We use $\gp$ to denote the group operation.  Consider the following $3$-query  dictatorship test for $f$:
\begin{enumerate}
	\item Sample $\V a = (a_1, a_2, \ldots, a_n)$ from $G^n$ uniformly at random.
	\item Sample $\V b = (b_1, b_2, \ldots, b_n)$ from $G^n$ uniformly at random.
	\item Calculate $\V c = (c_1, c_2, \ldots, c_n)$ such that $c_i = b_i^{-1} a_i^{-1}$.
	\item Check if $f(\V a) \gp f(\V b)\gp f(\V c) = 1_{G}$.
\end{enumerate}
Completeness is trivial: If $f$ is an $i^{th}$ dictator function, i.e., $f(x_1, x_2, \ldots, x_n) = x_i$, then the test passes with probability $1$. This is because we are essentially checking if $a_i\gp b_i \gp c_i = 1_{G}$, which is always true by the definition of $c_i$.

  We analyze the soundness of the test. The following lemma says that if the test passes with some non-trivial probability, then it must be the case that $f$ (or a minor variation of $f$) has a low-dimensional Fourier coefficient whose Hilbert-Schmidt norm is large.  The actual conclusion is somewhat stronger than this. In the next section, we will show that such a conclusion can be used to analyze the soundness of the final reduction (which is also presented in the next section).

\begin{lemma}
Assume $f$ is folded. For all $\eps>0$  and $\frac{1}{2}>\delta>0$, if $f$ passes the test with probability $\frac{1}{\left| \commutator{G}{G}\right|}+ \eps$, then there exist $\rho\in \irr(G)$ with $\dim(\rho)\geq 2$ and $1\leq i, j\leq \dim(\rho)$ such that for  $h (\V x) := \rho(f(\V x))_{ij}$ ,
$$ \max_{\substack{{\alpha\in\irr(G^n)},\\ \dim(\alpha)\geq 2, \\ \dimgeqi{\alpha}{2}< \frac{1}{2\delta^2}}.} \hsnorm{  {\hat{h}(\alpha)}}  \geq \frac{\eps}{|G|} - \delta.$$
\end{lemma}
\begin{proof}
Using Proposition~\ref{prop:sumdimsquare}, the probability that the test passes can be expressed as follows:
	\begin{align*}
		\Pr[\mbox{Test passes}] &= \frac{1}{|G|} \sum_{\rho\in \irr(G)} \dim(\rho)\E_{\V a, \V b, \V c} [\chi_\rho(f(\V a) \gp f(\V b) \gp f(\V c))]\\
		& = \frac{1}{|G|}\sum_{\substack{\rho\in \irr(G),\\ \dim(\rho) = 1}} \E_{\V a, \V b, \V c} [\chi_\rho(f(\V a) \gp f(\V b)\gp  f(\V c))] \\
		&\quad\quad\quad +  \frac{1}{|G|}\sum_{\substack{\rho\in \irr(G),\\ \dim(\rho) \geq 2}} \dim(\rho) \E_{\V a, \V b, \V c} [\chi_\rho(f(\V a) \gp f(\V b)\gp f(\V c))].
	\end{align*}
	In the first summation, for any $\rho\in \irr(G)$ such that $\dim(\rho) = 1$, using the multiplicativity of the characters, we have
	\begin{align*}
	\left|\chi_\rho(f(\V a) \gp f(\V b) \gp f(\V c))\right| & = \left|\chi_\rho(f(\V a)) \chi_\rho(f(\V b)) \chi_\rho(f(\V c))\right|\\
	& \leq \left| \chi_\rho(f(\V a)) \right|\cdot \left|\chi_\rho(f(\V b))\right|\cdot \left|\chi_\rho(f(\V c))\right|\\
	&= 1. \tag*{(unitary representations)}
	\end{align*}
	As the number of dimension $1$ representations of a group $G$ is equal to the size of the quotient  $G/\commutator{G}{G}$, we get
		\begin{equation}
		\label{eq:testpassing}
		\Pr[\mbox{Test passes}] \leq \frac{1}{\left| \commutator{G}{G}\right|}  +  \frac{1}{|G|}\sum_{\substack{\rho\in \irr(G),\\ \dim(\rho) \geq 2}} \dim(\rho) \E_{\V a, \V b, \V c} [\chi_\rho(f(\V a)\gp f(\V b)\gp f(\V c))].
	\end{equation}
	Now, fix any $\rho\in \irr(G)$ such that $\dim(\rho) \geq 2$. For $1\leq i,j\leq \dim(\rho)$, let $g_{ij} : G^n \rightarrow \C$ be defined as $g_{ij}(\V x) := \rho(f(\V x))_{ij}$. Using the definition of characters, we have
	
	\begin{align}
	\E_{\V a, \V b, \V c} [\chi_\rho(f(\V a)\gp  f(\V b)\gp f(\V c))] & =  	\E_{\V a, \V b, \V c} [\trace(\rho (f(\V a) \cdot f(\V b) \cdot f(\V c)))]\nonumber\\
\mbox{($\rho$ is a homomorphism)\quad\quad}	&= 	\E_{\V a, \V b, \V c} [\trace(\rho(f(\V a))\cdot \rho(f(\V b)) \cdot \rho(f(\V c)))]\nonumber\\
	& = \E_{\V a, \V b, \V c} \left[\sum_{1\leq i,j,k \leq \dim(\rho)} \rho(f(\V a))_{ij} \cdot \rho(f(\V b))_{jk} \cdot \rho(f(\V c))_{ki}\right]\nonumber\\
	& = \sum_{1\leq i,j,k \leq \dim(\rho)}  \E_{\V a, \V b, \V c} \left[g_{ij}(\V a)g_{jk}(\V b)g_{ki}(\V c )\right]\nonumber\\
	& = \sum_{1\leq i,j,k\leq\dim(\rho)} (g_{ij}\conv g_{jk}\conv g_{ki})(1_{G^n}). \label{eq:finalconv}
	\end{align}
	
%	\begin{remark}
%		Unlike~\cite{EngebretsenHR}, we are going to do decoding using one of the $(i,j,k)$. The analysis is much simpler to write as we move to functions mapping to $\C$ ( and not to matrices over $\C$ as was done in \cite{EngebretsenHR}).
%	\end{remark}
Since we assume that the test passes with probability $\frac{1}{\left| \commutator{G}{G}\right|}  +\eps$, from Equation~(\ref{eq:testpassing}) and Equation~(\ref{eq:finalconv}) (and using $\dim(\rho) \leq \sqrt{|G|}$), we conclude that there exists $\rho$ with $\dim(\rho)\geq 2$ and $1\leq i, j, k\leq \dim(\rho)$ such that 
$$ |(g_{ij}\conv g_{jk}\conv g_{ki})(1_{G^n})|\geq \frac{\eps}{|G|}.$$
We now analyze the term $(g_{ij}\conv g_{jk}\conv g_{ki})(1_{G^n})$ for a fixed $(i,j,k)$. For ease of notation, we write $h_1 := g_{ij}, h_2 := g_{jk}$ and $h_3 := g_{ki}$.
\begin{align*}
 \frac{\eps}{|G|}\leq |(g_{ij}\conv g_{jk}\conv g_{ki})(1_{G^n})| & = |(h_1\conv h_2\conv h_3)(1_{G^n})| \\
 &=\left| \sum_{\alpha\in \irr(G^n)} \dim(\alpha)\cdot \trace(\widehat{h_1\conv h_2\conv h_3}(\alpha))\right|\\
 &\leq  \sum_{\alpha\in \irr(G^n)} \dim(\alpha)\cdot |\trace(\hat{h}_1(\alpha) \hat{h}_2(\alpha)  \hat{h}_3(\alpha))|\\
    &= \sum_{\alpha\in \irr(G^n)} \dim(\alpha)\cdot \left|\langle \hat{h}_1(\alpha) \hat{h}_2(\alpha),  \ctranspose{\hat{h}_3(\alpha)}\rangle_{\End{\alpha_V}}\right|\\
    &\leq \sum_{\alpha\in \irr(G^n)} \dim(\alpha)\cdot \hsnorm{ \hat{h}_1(\alpha) \hat{h}_2(\alpha)}\hsnorm{  \ctranspose{\hat{h}_3(\alpha)}}.
\end{align*}

We now use Lemma~\ref{claim:folded_zero_coeff} to conclude that for all $1\leq r\leq 3$, $\hat{h}_r(\alpha)=0$ if $\dim(\alpha) = 1$.  Using this, we continue as follows:
\begin{align*}
|(h_1\conv h_2\conv h_3)(1_{G^n})| &\leq \sum_{\substack{\alpha\in \irr(G^n),\\ \dim(\alpha)\geq 2}} \dim(\alpha)\cdot \hsnorm{ \hat{h}_1(\alpha) \hat{h}_2(\alpha)}\hsnorm{  \ctranspose{\hat{h}_3(\alpha)}}\\
& = \sum_{\substack{{\alpha\in \irr(G^n)},\\ \dim(\alpha)\geq 2}} \dim(\alpha)\cdot \hsnorm{ \hat{h}_1(\alpha) \hat{h}_2(\alpha)}\hsnorm{\hat{h}_3(\alpha)}.\\
\end{align*}
 Let $D:= \frac{1}{2\delta^2}$. Now, we split the sum into two parts $|(h_1\conv h_2\conv h_3)(1_{G^n})|\leq \Theta_\low + \Theta_\high$ where 
$$\Theta_{\low} = \sum_{\substack{{\alpha\in \irr(G^n)},\\ \dim(\alpha)\geq 2, \\ \dimgeqi{\alpha}{2}< D}} \dim(\alpha)\cdot \hsnorm{ \hat{h}_1(\alpha) \hat{h}_2(\alpha)}\hsnorm{  {\hat{h}_3(\alpha)}},$$
and 
$$\Theta_{\high} = \sum_{\substack{{\alpha\in \irr(G^n)},\\ \dimgeqi{\alpha}{2}\geq D}} \dim(\alpha)\cdot \hsnorm{ \hat{h}_1(\alpha) \hat{h}_2(\alpha)}\hsnorm{  {\hat{h}_3(\alpha)}}.$$
\subsection{Bounding higher order terms}
In this section, we show that the high-degree terms can be upper bounded by a small constant, even though the three queries are perfectly correlated. 

We  bound $\Theta_{\high}$ as follows:
\begin{align*}
\Theta_{\high} &=  \sum_{\substack{{\alpha},\\ \dimgeqi{\alpha}{2}\geq D}} \dim(\alpha)\cdot \hsnorm{ \hat{h}_1(\alpha) \hat{h}_2(\alpha)}\hsnorm{  {\hat{h}_3(\alpha)}}\\ 
 &\leq  \sum_{\substack{{\alpha},\\ \dimgeqi{\alpha}{2}\geq D}} \dim(\alpha)\cdot \hsnorm{ \hat{h}_1(\alpha)}\hsnorm{ \hat{h}_2(\alpha)}\hsnorm{\hat{h}_3(\alpha)}\tag*{(Claim~\ref{claim:normteq})}\\
&\leq \frac{1}{\sqrt{2D}}\sum_{\substack{{\alpha},\\ \dimgeqi{\alpha}{2}\geq D}} \dim(\alpha)^{3/2}\cdot \hsnorm{ \hat{h}_1(\alpha)}\hsnorm{ \hat{h}_2(\alpha)}\hsnorm{\hat{h}_3(\alpha)}.
\end{align*}
Here, we used the fact that all the representations ${\alpha}$ of $G^n$ with $\dimgeqi{\alpha}{2} \geq D$ have dimensions at least $2^D\geq 2D$.  At this point, we would like to highlight the primary source of effectively bounding the higher-order terms. It is the size of $\dim(\alpha)$ in the summation. In Gowers'~\cite{Gowers08} proof, a similar expression appears in the analysis, with the same condition that all the representations in the summation have large dimension. It is in some sense the main difference between the abelian and the non-abelian setting (both in this work and Gowers'), similar to the Equation~(\ref{eq:convimprovement}) mentioned in the introduction.

Now, using Cauchy-Schwarz inequality,
\begin{align*}
\Theta_{\high} &\leq \frac{1}{\sqrt{2D}}\sum_{\substack{{\alpha},\\ \dimgeqi{\alpha}{2}\geq D}} \dim(\alpha)^{3/2}\cdot \hsnorm{ \hat{h}_1(\alpha)}\hsnorm{ \hat{h}_2(\alpha)}\hsnorm{\hat{h}_3(\alpha)}\\
&\leq \frac{1}{\sqrt{2D}} \Bigg( \sum_{\substack{{\alpha},\\ \dimgeqi{\alpha}{2}\geq D}} \hspace{-10pt} \dim(\alpha) \cdot \hsnorm{ \hat{h}_1(\alpha)}^2\Bigg)^{1/2} \cdot  \Bigg( \sum_{\substack{{\alpha},\\ \dimgeqi{\alpha}{2}\geq D}} \hspace{-10pt} \dim(\alpha)^2 \cdot \hsnorm{ \hat{h}_2(\alpha)}^2 \hsnorm{ \hat{h}_3(\alpha)}^2\Bigg)^{1/2} \\
&\leq \frac{1}{\sqrt{2D}} \Bigg( \sum_{\substack{{\alpha},\\ \dimgeqi{\alpha}{2}\geq D}} \hspace{-10pt} \dim(\alpha) \cdot \hsnorm{ \hat{h}_1(\alpha)}^2\Bigg)^{1/2} \cdot\\
& \quad \quad \quad  \left(\Bigg( \sum_{\substack{{\alpha},\\ \dimgeqi{\alpha}{2}\geq D}} \hspace{-10pt} \dim(\alpha) \cdot \hsnorm{ \hat{h}_2(\alpha)}^2 \Bigg) \cdot  \Bigg( \sum_{\substack{{\alpha},\\ \dimgeqi{\alpha}{2}\geq D}} \hspace{-10pt} \dim(\alpha) \cdot \hsnorm{ \hat{h}_3(\alpha)}^2\Bigg) \right)^{1/2}\\
& \leq \frac{1}{\sqrt{2D}} \cdot\|h_1\|_{L^2(G^n)} \cdot \|h_2\|_{L^2(G^n)} \cdot \|h_3\|_{L^2(G^n)}.\tag*{(Using Proposition~\ref{prop:parsevals})}
\end{align*}
Since $h_1$ was defined as $h_1(\V x) = \rho(f(\V x))_{ij}$ where $\rho \in \irr(G)$, $|h_1(\V x)|\leq 1$. The same is true for $h_2$ and $h_3$, and hence all the norms are bounded by $1$. Therefore,
$$\Theta_{\high} \leq \frac{1}{\sqrt{2D}}= \delta.$$

\subsection{Bounding lower order terms}
It remains to show that $\Theta_\low$ is related to the Fourier mass of $h_3$ on the low-dimensional representations. 
%This follows from ~\myqcite{Corollary 26}{EngebretsenHR} (the only change from\cite{EngebretsenHR} is that we moved to $h_1 : G^n \rightarrow \C$, but the proof of ~\myqcite{Corollary 26}{EngebretsenHR} is independent of this change). Here again, the analysis is much simpler. 

\begin{align*}
\Theta_{\low} &= \sum_{\substack{{\alpha},\\ \dim(\alpha)\geq 2, \\ \dimgeqi{\alpha}{2}< D}} \dim(\alpha)\cdot \hsnorm{ \hat{h}_1(\alpha) \hat{h}_2(\alpha)}\hsnorm{  {\hat{h}_3(\alpha)}}\\
&\leq  \max_{\substack{{\alpha},\\ \dim(\alpha)\geq 2, \\ \dimgeqi{\alpha}{2}< D}} \hsnorm{  {\hat{h}_3(\alpha)}} \cdot \Bigg(\sum_{\alpha} \dim(\alpha)\cdot \hsnorm{ \hat{h}_1(\alpha) \hat{h}_2(\alpha)}\Bigg).
\end{align*}
We can upper bound the summation by $1$ using the Cauchy-Schwarz inequality as follows:

\begin{align*}
\sum_{\alpha} \dim(\alpha)\cdot \hsnorm{ \hat{h}_1(\alpha) \hat{h}_2(\alpha)} &\leq \sum_{\alpha} \dim(\alpha)\cdot \hsnorm{ \hat{h}_1(\alpha)}\hsnorm{\hat{h}_2(\alpha)}\\
&\leq  \Bigg( \sum_{\alpha}  \dim(\alpha) \cdot \hsnorm{ \hat{h}_1(\alpha)}^2\Bigg)^{1/2}  \cdot  \Bigg( \sum_{\alpha} \dim(\alpha) \cdot \hsnorm{ \hat{h}_2(\alpha)}^2\Bigg)^{1/2}\\
&= \|h_1\|_2 \cdot \|h_2\|_2 \tag*{(Proposition~\ref{prop:parsevals})}\\
&\leq 1.
\end{align*}
Here, the last inequality uses the fact that $|h_1(\V x)|, |h_2(\V x)| \leq 1$ for all $\V x\in G^n$. Using the upper bound on $\Theta_\high$, we have $\Theta_\low \geq \frac{\eps}{|G|}-\delta$. Therefore, we get
$$   \max_{\substack{{\alpha},\\ \dim(\alpha)\geq 2, \\ \dimgeqi{\alpha}{2}< D}} \hsnorm{  {\hat{h}_3(\alpha)}}  \geq \left(\frac{\eps}{|G|}-\delta\right).$$

%This can be translated into a decoding strategy with success probability $\approx O_G\left(\frac{\Theta_\low^2}{D}\right)$. 
%\paragraph{Decoding  strategy:} Player 1 picks $\rho\in \irr(G)$ and $1\leq i, j \leq  \dim(\rho)$ uniformly at random. The player then picks ${\alpha}$ with probability proportional to $\dim(\alpha)\hsnorm{ \hat{g}_{ij}(\alpha)}^2$. Let $\alpha = \otimes_{i=1}^n {\rho_i}$. She then outputs an index $i\in [n]$ such that $\dim(\rho_i) \geq 2$ uniformly at random. 

%Player 2 picks $\rho\in \irr(G)$ and $1\leq j, k\leq \dim(\rho)$ uniformly at random. The player then picks $\alpha'\in \irr(G^n)$ with probability proportional to $\dim(\alpha')\hsnorm{ \hat{g}_{jk}(\alpha')}^2$. Let $\alpha' = \otimes_{i=1}^n \rho'_i$. She then outputs an index $i\in [n]$ uniformly at random such that $\dim(\rho'_i) \geq 2$ . 
\end{proof}

\section{Main Reduction}

\newcommand{\termm}[3]{\mathbf{Term}^{#1}(#2, #3)}

\newcommand{\ff}{{F^{ik}_{\alpha}}}
\newcommand{\ti}{t}
\newcommand{\tj}{t'}

In this section, we prove Theorem~\ref{thm:mainnphardness}. We give a
reduction from an instance of a Layered {\LC}, $\calH=(\calV = \{V_1, V_2, \ldots, V_\ell\},\{\Pi_{\ti,\tj}\}_{1\leq \ti<\tj\leq \ell},\{[R_\ti]\}_{\ti\in [\ell]})$  as in Definition~\ref{def:layered-LC}, to
a $3$-\LIN\ instance $\calI$ over a non-abelian group $G$. Towards obtaining the $1$ vs. $\frac{1}{\left|\commutator{G}{G}\right|} + \delta$ hardness of Max-$3$-\LIN\ for $\delta\in \left(0,\frac{1}{2}\right)$, we will use the following setting of $\ell$ and $r$ in Theorem~\ref{theorem:multiLC}
 $$ 2^{-r}\leq \min\left\{ \frac{\delta^{10}}{(2|G|)^{100}}, \frac{\delta^2}{10|G|^{10C}}\right\}, \quad \ell \geq \left(\frac{8|G|^3}{\delta}\right)^4,$$
where $C := \frac{8|G|^6}{\delta^2}$.

The reduction is given in Figure~\ref{fig:mainreduction}.

\begin{figure}
\begin{tcolorbox}[colback=white!40!white]
The set of variables of  $\calI$ 
is $\bigcup_{\ti\in [\ell]} V_\ti \times G^{[R_\ti]}$. Any assignment to the instance $\calI$ is given by
a set of functions $f_{v} : G^{[R_\ti]}\rightarrow G$ for each
$v\in V_\ti$ for all $\ti\in [\ell]$. We further enforce that these functions are {\em folded} (the set $G^{[R_t]}$ can be thought of as $G^{R_t-1}$ variables and their literals by the following association: every $\V x\in G^{[R_t]}$ such that $x_1=1_G$ is a variable, and its $g$-literal is given by $g\V x$).\\

The distribution of the $3$-\LIN\ constraints in $\calI$ is given by the following procedure:
\begin{enumerate}
	\item Pick a uniformly random pair $(\ti, \tj)$ satisfying $1\leq \ti <\tj \leq \ell$.
	\item Choose a constraint  $\pi_{vu} \in \Pi_{\ti,\tj}$ of $\calH$ uniformly at random.
	\item Sample $\V a = (a_1, a_2, \ldots, a_{R_\ti})$ from $G^{R_\ti}$ uniformly at random.
	\item Sample $\V b = (b_1, b_2, \ldots, b_{R_{\tj}})$ from $G^{R_\tj}$ uniformly at random.
	\item  Let $\V c= (c_1, c_2, \ldots, c_{R_\ti})$ be such that $c_i = (b\circ \pi_{vu})_i^{-1}\gp  a_i^{-1}$, here ${\V x }\circ \pi\in G^{R_\ti}$ is the string defined as $(x \circ \pi)_i := x_{\pi(i)}$
	for $i \in [R_\ti]$.
	\item Output the constraint $f_v(\V a)\gp f_u(\V b)\gp f_v(\V c) = 1_G$.
	
\end{enumerate}
\end{tcolorbox}
    \caption{Reduction from a Layered \LC\ instance \\$\calH=(\calV =\{V_{\ti}\}_{\ti\in [\ell]},\{\Pi_{\ti,\tj}\}_{1\leq \ti<\tj\leq \ell},\{[R_\ti]\}_{\ti\in [\ell]})$ to a Max-$3$-LIN $\calI$ over $G$.}
    \label{fig:mainreduction}
\end{figure}
The value of the instance $\val(\calI)$ is the maximum probability that the above test is satisfied, where the maximum is over all folded functions $\{f_v\}_{v\in \cup_t V_t}$.

\subsection{Analysis}

\begin{lemma}[Completeness]
	\label{lemma:completeness_nphard}
	If $\calH$ is a satisfiable instance   of the Layered {\LC}, then  $\val(\calI) = 1$.
\end{lemma}
\begin{proof}
	Fix a satisfying assignment $\{\sigma_\ti : V_\ti\rightarrow [R_\ti]\}_{\ti\in [\ell]}$ of $\calH$. Consider the long code encodings of the labeling $\sigma_\ti$ : for every $\ti \in [\ell]$ and $v\in V_\ti$, set $f_v(\V x) = x_{\sigma_\ti(v)}$. We show that this assignment to $\calI$ satisfies all the constraints. For the constraint $\pi_{vu} \in \Pi_{\ti,\tj}$, the LHS of the test becomes
	\begin{align*}
	f_v(\V a)\gp f_u(\V b) \gp f_v(\V c) &= a_{\sigma_\ti(v)}\gp b_{\sigma_{\tj}(u)} \gp c_{\sigma_\ti(v)}\\
	& = a_{\sigma_{\ti}(v)}\gp b_{\sigma_{\tj}(u)} \gp (b\circ \pi_{vu})_{\sigma_{\ti}(v)}^{-1}\gp a_{\sigma_{\ti}(v)}^{-1}\\
	& = a_{\sigma_{\ti}(v)}\gp b_{\sigma_{\tj}(u)} \gp b_{\pi_{vu}(\sigma_{\ti}(v))}^{-1}\gp a_{\sigma_{\ti}(v)}^{-1}\\
	& = a_{\sigma_{\ti}(v)}\gp b_{\sigma_{\tj}(u)} \gp b_{\sigma_{\tj}(u)}^{-1}\gp a_{\sigma_{\ti}(v)}^{-1} \tag*{$(\pi_{vu}(\sigma_{\ti}(v)) = \sigma_{\tj}(u))$}\\
	& = 1_G.
	\end{align*}
Therefore, the assignment to the $3$-\LIN\ instance $\calI$ satisfies all the constraints. 
\end{proof}

We now prove the main soundness lemma which is as follows
\begin{lemma}[Soundness]
	\label{lemma:soundness_nphard} 
	If $\calH$ is at most $ 2^{-r}$-satisfiable, then $\val(\calI) \leq \frac{1}{\left| \commutator{G}{G}\right|} + \delta$.
\end{lemma}

Before we see the proof, we first conclude that this is enough to prove the main theorem of this paper.

\begin{proof}[Proof of Theorem~\ref{thm:mainnphardness}]
Lemma~\ref{lemma:completeness_nphard} and Lemma~\ref{lemma:soundness_nphard} above, along with the $\np$-hardness of the Layered \LC\ from Theorem~\ref{theorem:multiLC} imply our main theorem, Theorem~\ref{thm:mainnphardness}, for any constant $\eps>0$.
\end{proof}

\begin{proof}[Proof of Lemma~\ref{lemma:soundness_nphard}]
Fix any assignment to the instance $\calI$ given by
a set of functions $f_v : G^{R_\ti} \rightarrow G$ for each $\ti\in [\ell]$ and 
$v\in V_\ti$.	Using Proposition~\ref{prop:sumdimsquare}, the value of the instance for this assignment is given by:
\begin{align*}
\E_{1\leq \ti<\tj\leq \ell}~~\E_{\pi_{vu}\in \Pi_{\ti,\tj}}~~\E_{\V a, \V b, \V c} \left[ \frac{1}{|G|} \sum_{\rho\in \irr(G)} \dim(\rho)\chi_\rho(f_v(\V a) \gp f_u(\V b)\gp f_v(\V c))\right]\\
\end{align*}
For any $\rho\in \irr(G)$ such that $\dim(\rho)= 1$, we have
\begin{align*}
\left|\chi_\rho(f_v(\V a)\gp f_u(\V b)\gp f_v(\V c))\right| & = \left|\chi_\rho(f_v(\V a)) \cp \chi_\rho(f_u(\V b))\cp \chi_\rho(f_v(\V c))\right|\\
& \leq \left| \chi_\rho(f_v(\V a)) \right|\cp \left|\chi_\rho(f_u(\V b))\right|\cp \left|\chi_\rho(f_v(\V c))\right|\\
&= 1. \tag*{(unitary representations)}
\end{align*}
As the number of dimension $1$ representations of a group $G$ is equal to the size of the quotient group $G/\commutator{G}{G}$, we get
\begin{align*}
\val(\calI) &\leq   \frac{1}{\left| \commutator{G}{G}\right|}  + \frac{1}{|G|}\left| \E_{1\leq \ti<\tj\leq \ell}~~\E_{\pi_{vu}\in \Pi_{\ti,\tj}}~~  \left[ \sum_{\substack{\rho\in \irr(G), \\ \dim(\rho)\geq 2}} \dim(\rho)\E_{\V a, \V b, \V c} [\chi_\rho(f_v(\V a) \gp f_u(\V b)\gp f_v(\V c))]\right]\right|\\
 &\leq   \frac{1}{\left| \commutator{G}{G}\right|}  +   \sum_{\substack{\rho\in \irr(G), \\ \dim(\rho)\geq 2}} \left|\E_{1\leq \ti<\tj\leq \ell}~~\E_{\pi_{vu}\in \Pi_{\ti,\tj}}~~\E_{\V a, \V b, \V c} \left[\chi_\rho(f_v(\V a) \gp f_u(\V b)\gp f_v(\V c))\right]\right|.
\end{align*}

The lemma follows from the following Claim~\ref{claim:onetermbound}.
\end{proof}%end soundness claim

\begin{claim}
	\label{claim:onetermbound}
		If $\calH$ is at most $2^{-r}$-satisfiable, then for every $\rho\in \irr(G)$ such that $\dim(\rho)\geq 2$,
	$$\left|\E_{1\leq \ti<\tj\leq \ell}~~\E_{\pi_{vu}\in \Pi_{\ti,\tj}}~~ \E_{\V a, \V b, \V c} [\chi_\rho(f_v(\V a) \gp f_u(\V b)\gp f_v(\V c))]\right|\leq \frac{\delta}{|G|}.$$
\end{claim}
\begin{proof} 
Fix any $\rho\in \irr(G)$ such that $\dim(\rho)\geq 2$. Let 
\begin{align*}
    \Gamma &:= \E_{1\leq \ti<\tj\leq \ell}~~\E_{\pi_{vu}\in \Pi_{\ti,\tj}}~~  \E_{\V a, \V b, \V c} [\chi_\rho(f_v(\V a) \gp f_u(\V b)\gp f_v(\V c))]\\
    &=\E_{1\leq \ti<\tj\leq \ell}[\Gamma_{\ti,\tj}].
\end{align*}
where 
$$\Gamma_{\ti,\tj}:=\E_{\pi_{vu}\in \Pi_{\ti,\tj}}~~  \E_{\V a, \V b, \V c} [\chi_\rho(f_v(\V a) \gp f_u(\V b)\gp f_v(\V c))].$$

Fix $1\leq \ti<\tj\leq \ell$ and $\pi_{vu} \in \Pi_{\ti,\tj}$. We first look at the inner expectation. For $1\leq \ii,\jj\leq \dim(\rho)$, let $g_{\ii\jj} : G^{R_\ti} \rightarrow \C$ be defined as $g_{\ii\jj}(\V x) := \rho(f_v(\V x))_{\ii\jj}$. Also, let $h_{\ii\jj} : G^{R_{\tj}}\rightarrow \C$ be defined as $h_{\ii\jj}(\V y) := \rho(f_u(\V y))_{\ii\jj}$. We have

\begin{align*}
\E_{\V a, \V b, \V c} [\chi_\rho(f_v(\V a) f_u(\V b) f_v(\V c))] & =  	\E_{\V a, \V b, \V c} [\trace(\rho (f_v(\V a) \gp f_u(\V b) \gp f_v(\V c)))]\\
\mbox{($\rho$ is a homomorphism)\quad\quad}	&= 	\E_{\V a, \V b, \V c} [\trace(\rho(f_v(\V a)) \cp \rho(f_u(\V b)) \cp \rho(f_v(\V c)))]\\
& = \E_{\V a, \V b, \V c} \left[\sum_{1\leq \ii,\jj,\kk \leq \dim(\rho)} \rho(f_v(\V a))_{\ii\jj} \cp \rho(f_u(\V b))_{\jj\kk} \cp \rho(f_v(\V c))_{\kk\ii}\right]\\
& = \sum_{1\leq \ii,\jj,\kk \leq \dim(\rho)}  \E_{\V a, \V b, \V c} \left[g_{\ii\jj}(\V a)\cp h_{\jj\kk}(\V b)\cp g_{\kk\ii}(\V c )\right].
\end{align*}
 We now analyze the term $\Theta^{(v,u)}_{\ii,\jj,\kk} := \E_{\V a, \V b, \V c} \left[g_{\ii\jj}(\V a)h_{\jj\kk}(\V b)g_{\kk\ii}(\V c )\right]$ for a fixed $(\ii,\jj,\kk)$. For ease of notation, we write $g' := g_{\ii\jj}, h := h_{\jj\kk}$ and $g := g_{\kk\ii}$. Also, we use $\pi$ for $\pi_{vu}$. The following identity follows from Proposition~\ref{prop:convolution} and we will use it to simplify $\Theta^{(v,u)}_{\ii,\jj,\kk}$. 

$$\E_{\V a, \V b}[g'(\V a) \cp g((\V b\circ \pi)^{-1}\gp \V a^{-1})] = \E_{\V b}[(g\conv g')((\V b\circ \pi)^{-1})].$$
We now bound the expectation as follows.
\begin{align*}
\Theta^{(v,u)}_{\ii,\jj,\kk} &=\E_{\V a, \V b, \V c} \left[g_{\ii\jj}(\V a)\cp h_{\jj\kk}(\V b)\cp g_{\kk\ii}(\V c )\right] \\
& =  \E_{\V a, \V b, \V c} \left[g'(\V a)\cp h(\V b)\cp g((\V b\circ \pi)^{-1}\gp \V a^{-1} )\right]\\
& =  \E_{\V b} \left[(g\conv g')(\V b^{-1}\circ \pi)\cp h(\V b)\right]\\
& =  \E_{\V b} \left[(g\conv g')(\V b\circ \pi)\cp h(\V b^{-1})\right]\\
& =  \E_{\V b} \left[\left(\sum_{\alpha} \dim(\alpha) \trace(\hat{g}(\alpha)\hat{g'}(\alpha) \alpha(\V b\circ \pi))\right) \cp \left(\sum_{\beta} \dim(\beta) \trace(\hat{h}(\beta) \beta(\V b^{-1}))\right)\right]\\
& =  \sum_{\substack{\alpha, \beta , \\ \dim(\alpha), \dim(\beta)\geq 2}} \underbrace{\dim(\alpha)\dim(\beta) \E_{\V b} \left[\trace(\hat{g}(\alpha)\hat{g'}(\alpha)  {\alpha}(\V b\circ \pi)) \cp  \trace(\hat{h}(\beta)  {\beta}(\V b^{-1}))\right]}_{\termm{(v,u)}{\alpha}{\beta}},
\end{align*}
where the last step uses the fact that the functions $g, g'$ and $h$ satisfy the condition of Lemma~\ref{claim:folded_zero_coeff} and hence $\hat{g}(\alpha) = 0$ if $\dim(\alpha)  = 1$ (same for $\hat{g'}(\alpha)$ and $\hat{h}(\beta)$ ).

We now break the sum into two parts for $C := \frac{8|G|^6}{\delta^2}$:

\begin{align*}
\Theta^{(v,u)}_{\ii,\jj,\kk} (\low) : = \sum_{\substack{\alpha, \beta , \\ \dim(\alpha), \dim(\beta)\geq 2, \\\dimgeqi{\alpha}{2} \leq  C}} {\termm{(v,u)}{\alpha}{\beta}}, \quad\quad\quad \Theta^{(v,u)}_{\ii,\jj,\kk} (\high): = \sum_{\substack{\alpha, \beta,  \\ \dim(\alpha), \dim(\beta)\geq 2, \\ \dimgeqi{\alpha}{2} >  C}} {\termm{(v,u)}{\alpha}{\beta}}.
\end{align*}
Recall, $\dimgeqi{\alpha}{2}$ denotes the number of representations in $\alpha = (\rho_1, \rho_2,\ldots, \rho_{R_t})$ which are of dimensions at least $2$. With these notations, we have

$$\Gamma_{\ti,\tj} := \sum_{\ii,\jj,\kk}  \left(\E_{\pi_{vu}\in \Pi_{\ti,\tj}}\left[ \Theta^{(v,u)}_{\ii,\jj,\kk} (\low)\right] +  \E_{\pi_{vu}\in \Pi_{\ti,\tj}}\left[\Theta^{(v,u)}_{\ii,\jj,\kk} (\high)\right] \right).$$
The upper bound on the absolute value of $\Gamma :=\E_{1\leq \ti<\tj\leq \ell}[\Gamma_{\ti,\tj}]$ follows from Claim~\ref{claim:lowterms}, Claim~\ref{claim:highterms} and triangle inequality (and also noting that $\ii, \jj$ and $\kk$ take at most $\sqrt{|G|}$ distinct values).
\end{proof}%end each term in success prob is small

\begin{claim}
	\label{claim:lowterms}
	If $\calH$ is at most $\frac{\delta^2}{10|G|^{10C}}$-satisfiable,  then for every $1\leq \ti<\tj\leq \ell$, $\rho\in \irr(G)$ such that $\dim(\rho)\geq 2$, and every $1\leq\ii,\jj,\kk\leq \dim(\rho)$, 
	$$\left| \E_{\pi_{vu}\in \Pi_{\ti,\tj}}\left[\Theta^{(v,u)}_{\ii,\jj,\kk}(\low)\right] \right|\leq \frac{\delta}{2|G|^3}.$$
\end{claim}

\begin{claim} Suppose $\calH$ is at most $\frac{\delta^{10}}{(2|G|)^{100}}$-satisfiable. Then for every $\rho\in \irr(G)$ such that $\dim(\rho)\geq 2$ and $1\leq\ii,\jj,\kk\leq \dim(\rho)$, 
	\label{claim:highterms}
	$$\left|\E_{1\leq \ti<\tj\leq \ell}~~\E_{\pi_{vu}\in \Pi_{\ti,\tj}}\left[\Theta^{(v,u)}_{\ii,\jj,\kk} (\high)\right] \right| \leq \frac{\delta}{2|G|^3}.$$
\end{claim}

\subsubsection{Bounding the \texorpdfstring{$\low$}{low} terms}

\begin{claim} (Restatement of Claim~\ref{claim:lowterms})
	If $\calH$ is at most $\frac{\delta^2}{10|G|^{10C}}$-satisfiable,  then for every $1\leq \ti<\tj\leq \ell$, $\rho\in \irr(G)$ such that $\dim(\rho)\geq 2$, and every $1\leq\ii,\jj,\kk\leq \dim(\rho)$, 
	$$\left| \E_{\pi_{vu}\in \Pi_{\ti,\tj}}\left[\Theta^{(v,u)}_{\ii,\jj,\kk}(\low)\right] \right|\leq \frac{\delta}{2|G|^3}.$$
\end{claim}

\begin{proof}
Fix any $1\leq \ti<\tj\leq \ell$ and $\rho\in \irr(G)$ such that $\dim(\rho)\geq 2$.
Assume towards contradiction that there exists $1\leq \ii,\jj,\kk \leq \dim(\rho)$ such that 
$$\E_{\pi_{vu}\in \Pi_{\ti, \tj}}\left[\left| \Theta^{(v,u)}_{\ii,\jj,\kk}(\low)\right|\right]  \geq \left| \E_{\pi_{vu}\in \Pi_{\ti, \tj}}\left[\Theta^{(v,u)}_{\ii,\jj,\kk}(\low)\right]  \right|> \frac{\delta}{2|G|^3}.$$ 

We show that in this case, $\calH$ has a $>\frac{\delta^2}{10|G|^{10C}}$ satisfying assignment which is a contradiction. Consider the term $\termm{(v, u)}{\alpha}{\beta}$ when $\alpha = (\rho_1, \rho_2,\ldots, \rho_{R_{\ti}})$ and $\beta = (\tau_1, \tau_2, \ldots, \tau_{R_{\tj}})$.

	\begin{align*}
\termm{(v,u)}{\alpha}{\beta} & =  \dim(\alpha)\dim(\beta) \E_{\V b} \left[\trace(\hat{g}(\alpha)\hat{g'}(\alpha)  {\alpha}(\V b\circ \pi)) \cp  \trace(\hat{h}(\beta)  {\beta}(\V b^{-1}))\right]\\
	& = \dim(\alpha)\dim(\beta) \E_{\V b} \left[\sum_{i,k} (\hat{g}(\alpha)\hat{g'}(\alpha))_{ik}  {\alpha}(\V b\circ \pi)_{ki} \cp \sum_{i',k'}\hat{h}(\beta)_{i'k'}  {\beta}(\V b^{-1})_{k'i'}\right]\\
	& =  \dim(\alpha)\dim(\beta) \E_{\V b} \left[\sum_{\substack{i,k\\ i',k'}} (\hat{g}(\alpha)\hat{g'}(\alpha))_{ik}  {\alpha}(\V b\circ \pi)_{ki} \hat{h}(\beta)_{i'k'}  {\beta}(\V b^{-1})_{k'i'}\right]\\
	& =  \dim(\alpha)\dim(\beta) \sum_{\substack{i,k\\ i',k'}} (\hat{g}(\alpha)\hat{g'}(\alpha))_{ik}  \hat{h}(\beta)_{i'k'}  \E_{\V b} \left[  {\alpha}(\V b\circ \pi)_{ki} \cp  {\beta}(\V b^{-1})_{k'i'}\right],
	\end{align*}
	where $(i,k)$ are the tuples $i = (i_1, i_2, \ldots, i_{R_{\ti}})$ and $k = (k_1, k_2, \ldots, k_{R_{\ti}})$ such that for all $\ell\in [R_{\ti}]$, $1\leq  i_\ell, k_\ell\leq \dim(\rho_\ell)$. Similarly, $(i',k')$ are the tuples $i' = (i'_1, i'_2, \ldots, i'_{R_{\tj}})$ and $k' = (k'_1, k'_2, \ldots, k'_{R_{\tj}})$ such that for all $\ell'\in [R_{t'}]$, $1\leq  i'_{\ell'}, k'_{\ell'}\leq \dim(\tau_{\ell'})$. Now, 
	\begin{align*}	
	\E_{\V b}  \left[  {\alpha}(\V b\circ \pi)_{ki} \cp   {\beta}(\V b^{-1})_{k'i'}\right] & = \E_{\V b}  \left[ \prod_{{\ell'} = 1}^{R_{\tj}}  \left({\tau_{\ell'}}(b_{\ell'}^{-1})_{k'_{\ell'}i'_{\ell'}} \cdot \prod_{\ell \in \pi^{-1}({\ell'})}  {\rho_{\ell}}(b_{\ell'})_{k_{{\ell}}i_{{\ell}}}\right) \right] \\
	& =  \prod_{{\ell'} = 1}^{R_{\tj}} \E_{ b}  \left[  {\tau_{\ell'}}(b^{-1})_{k'_{\ell'}i'_{\ell'}} \cdot \prod_{\ell \in \pi^{-1}({\ell'})}  {\rho_{\ell}}(b)_{k_{\ell}i_{\ell}} \right].
	\end{align*}
	Now suppose there exists an ${\ell'}$ such that for all $\ell\in \pi^{-1}(\ell')$, $\dim(\rho_{\ell})=1$. Since the product of dimension $1$ representations is also a dimension $1$ representation, for the expectation to be nonzero, $\dim(\tau_{\ell'})$ must be $1$, by Proposition~\ref{prop:orthmatrixentries}. Thus, if $\dim(\beta)\geq 2$, then there must exist an $\ell'$ such that $\dim(\tau_{\ell'})\geq 2$ and an $\ell\in  \pi^{-1}(\ell')$ such that $\dim(\rho_\ell)\geq 2$ (again, for the expectation to be non-zero). Therefore, we conclude that the terms that are nonzero in the expression $\Theta^{(v,u)}_{\ii,\jj,\kk} (\low)$ are all $(\alpha, \beta)$ such that for all $\ell' \in [R_{\tj}]$ whenever $\dim(\tau_{\ell'})\geq 2$, there exists an $\ell\in \pi^{-1}(\ell')$ such that $\dim(\rho_{\ell})\geq 2$. Let us define $\pi_{\geq 2}(\alpha) = \{\ell' \in [R_{\tj}] \mid \exists \ell\in \pi^{-1}(\ell'), \dim(\rho_{\ell})\geq 2\}$ and $\beta_{\geq 2} = \{ \ell'\in [R_{\tj}] \mid \dim(\tau_{\ell'})\geq 2\}$. We get
	\begin{align*}
	\Theta^{(v,u)}_{\ii,\jj,\kk}(\low) & = \sum_{\substack{\alpha , \beta, \\ \dim(\alpha) \geq 2, \dim(\beta)\geq 2, \\  \dimgeqi{\alpha}{2} \leq C, \\ \beta_{\geq 2} \subseteq \pi_{\geq 2}(\alpha)}} \dim(\alpha)\dim(\beta) \sum_{i,k} (\hat{g}(\alpha)\hat{g'}(\alpha))_{ik} \sum_{i',k'} \hat{h}(\beta)_{i'k'} \E_{\V b}  \left[  {\alpha}(\V b\circ \pi)_{ki} \cp   {\beta}(\V b^{-1})_{k'i'}\right].
	\end{align*}

	Now, let $F^{ik}_{\alpha} : G^{R_{\tj}} \rightarrow \C$ be the following function:
	$$\ff(\V b) :=  {\alpha}(\V b^{-1}\circ \pi)_{ki}.$$
	Note that,
	\begin{equation}
	\label{eq:sum_one}
	\sum_{k} \|\ff\|_2^2 = \sum_{k}  \E_{\V b}[ |{\alpha}(\V b^{-1}\circ \pi)_{ki}|^2] = \E_{\V b} \sum_{k}  |{\alpha}(\V b^{-1}\circ \pi)_{ki}|^2 = 1,
	\end{equation}
	where the last equality uses the fact that the sum expression is exactly the squared norm of the column $i$ of representation $ {\alpha}$, which is 1 ($ {\alpha}(.)$ is unitary).
	We now analyze the expectation.
	\begin{align*}
	\E_{\V b} \left[  {\alpha}(\V b\circ \pi)_{ki} \cp  {\beta}(\V b^{-1})_{k'i'}\right] &= \E_{\V b} \left[ \ff(\V b^{-1})\cp  {\beta}(\V b^{-1})_{k'i'}\right]\\
	& = \E_{\V b} \left[\sum_{\beta'} \dim(\beta') \trace(\hat{\ff}(\beta')  {\beta'}(\V b^{-1})^\star) \cp  {\beta}(\V b^{-1})_{k'i'}\right]\\
	& = \E_{\V b} \left[\sum_{\beta'} \dim(\beta') \trace(\hat{\ff}(\beta')  {\beta'}(\V b)) \cp  {\beta}(\V b^{-1})_{k'i'}\right]\\
	& = \E_{\V b} \left[\sum_{\beta'} \dim(\beta') \sum_{i'',k''} \hat{\ff}(\beta')_{i'',k''}  {\beta'}(\V b)_{k'',i''} \cp  {\beta}(\V b^{-1})_{k'i'}\right]\\
	& = \sum_{\beta'} \dim(\beta') \sum_{i'',k''} \hat{\ff}(\beta')_{i'',k''} \E_{\V b} \left[  {\beta'}(\V b)_{k'',i''} \cp  {\beta}(\V b^{-1})_{k'i'}\right].
	\end{align*}
	By Proposition~\ref{prop:orthmatrixentries}, the expectation is $1/\dim(\beta')$ if $\beta' = \beta$, $i'' = k'$ and $k'' = i'$; otherwise, it is zero. Therefore,
	\begin{align*}
	\E_{\V b} \left[  {\alpha}(\V b\circ \pi)_{ki} \cp  {\beta}(\V b^{-1})_{k'i'}\right] 	& = \hat{\ff}(\beta)_{k',i'}.
	\end{align*}
	Plugging this  into $\Theta^{(v,u)}_{\ii,\jj,\kk}(\low)$, we get (suppressing the conditions on $\alpha, \beta$ in the summation for now)
	\begin{align*}
	\left|\Theta^{(v,u)}_{\ii,\jj,\kk}(\low)\right|^2 & = \left|\sum_{\alpha , \beta} \dim(\alpha)\dim(\beta) \sum_{\substack{i,k,\\ i',k'}} (\hat{g}(\alpha)\hat{g'}(\alpha))_{ik}  \hat{h}(\beta)_{i'k'} \hat{\ff}(\beta)_{k',i'}\right|^2\\
	& = \left|\sum_{\alpha , \beta} \dim(\alpha)\dim(\beta) \sum_{\substack{i,j, k,\\ i',k'}} \hat{g}(\alpha)_{ij} \hat{g'}(\alpha)_{jk}  \hat{h}(\beta)_{i'k'} \hat{\ff}(\beta)_{k',i'}\right|^2\\
	& \leq \left(\sum_{\alpha , \beta} \dim(\alpha)\dim(\beta) \sum_{\substack{i,j, k,\\ i',k'}} |\hat{g'}(\alpha)_{jk}|^2  |\hat{h}(\beta)_{i'k'}|^2 \right) \left(  \sum_{\alpha , \beta} \dim(\alpha)\dim(\beta) \sum_{\substack{i,j, k,\\ i',k'}}  |\hat{g}(\alpha)_{ij}|^2 |\hat{\ff}(\beta)_{k',i'}|^2 \right).
	\end{align*}
	We can bound the second term as follows.
	\begin{align*}
	\left(  \sum_{\alpha, \beta} \dim(\alpha)\dim(\beta) \sum_{\substack{i,j, k,\\ i',k'}}  |\hat{g}(\alpha)_{ij}|^2 |\hat{\ff}(\beta)_{k',i'}|^2 \right) & = \left(  \sum_{\alpha} \dim(\alpha) \sum_{i,j}  |\hat{g}(\alpha)_{ij}|^2 \sum_{k}\sum_{\beta} \dim(\beta) \sum_{i',k'} |\hat{\ff}(\beta)_{k',i'}|^2 \right) \\
	& = \left(  \sum_{\alpha} \dim(\alpha) \sum_{i,j}  |\hat{g}(\alpha)_{ij}|^2 \sum_{k} \|\ff\|_2^2 \right) \\
	& = \left(  \sum_{\alpha} \dim(\alpha) \sum_{i,j}  |\hat{g}(\alpha)_{ij}|^2 \right) \tag*{(Using Equation~(\ref{eq:sum_one}))}\\
	& = \|g\|_2^2 \leq 1.
	\end{align*}
	Therefore,
	\begin{align*}
	\left|\Theta^{(v,u)}_{\ii,\jj,\kk}(\low)\right|^2 & \leq \sum_{\substack{\alpha , \beta, \\ \dim(\alpha) , \dim(\beta)\geq 2,\\  \dimgeqi{\alpha}{2} \leq C,\\ \beta_{\geq 2} \subseteq \pi_{\geq 2}(\alpha)}}\dim(\alpha)\dim(\beta) \sum_{\substack{i,j, k, \\ i',k'}} |\hat{g'}(\alpha)_{jk}|^2  |\hat{h}(\beta)_{i'k'}|^2 \\
	& \leq |G|^C \sum_{\substack{\alpha , \beta,\\ \dim(\alpha) , \dim(\beta)\geq 2,\\  \dimgeqi{\alpha}{2} \leq C, \\ \beta_{\geq 2} \subseteq \pi_{\geq 2}(\alpha)}}\dim(\alpha)\dim(\beta) \sum_{\substack{j, k, \\ i',k'}} |\hat{g'}(\alpha)_{jk}|^2  |\hat{h}(\beta)_{i'k'}|^2 \\
	& = |G|^C \sum_{\substack{\alpha , \beta, \\ \dim(\alpha) , \dim(\beta)\geq 2, \\  \dimgeqi{\alpha}{2} \leq C, \\ \beta_{\geq 2} \subseteq \pi_{\geq 2}(\alpha)}}\dim(\alpha)\dim(\beta) \hsnorm{\hat{g'}(\alpha)}^2  \hsnorm{\hat{h}(\beta)}^2 .
	\end{align*}
Here, in the second inequality, the omitted sum over $i$ contributes $\dim(\alpha)$, while $\dimgeqi{\alpha}{2}\leq C$ implies $\dim(\alpha)\leq \sqrt{|G|}^{C}\leq |G|^C$. We now show that if $\Theta^{(v,u)}_{\ii,\jj,\kk}(\low)$ is large for a typical $\pi_{vu} \in \Pi_{\ti,\tj}$, then it can be used to get a good labeling for the Label Cover instance $\calH$. 

\paragraph{Randomized labeling.} Consider the following randomized labeling. For each $v\in V_\ti$, consider $g_{\ii\jj} : G^{R_{\ti}} \rightarrow \C$ which is defined as $g_{\ii\jj}(\V x) := \rho(f_v(\V x))_{\ii\jj}$. Select $\alpha = (\rho_1, \rho_2,\ldots, \rho_{R_{\ti}})$  with probability $\dim(\alpha) \hsnorm{\hat{g}_{\ii\jj}(\alpha)}^2$.\footnote{As $\|g_{\ii\jj}\|_2^2\leq 1$, by Parseval's identity, the sum of these probabilities is at most $1$. Select $\bot$ with the remaining probability to make it a distribution.} Select a uniformly random $\ell_v \in [R_{\ti}]$ such that $\dim(\rho_{\ell_v})\geq 2$ and assign the label $\ell_v$ to $v$. If there is no such index, then assign $\bot$.

For each $u\in V_{\tj}$, let $h_{\jj\kk} : G^{R_{\tj}} \rightarrow \C$ be defined as $h_{\jj\kk}(\V y) := \rho(f_u(\V y))_{\jj\kk}$. Select $\beta = (\tau_1, \tau_2,\ldots, \tau_{R_{\tj}})$  with probability $\dim(\beta) \hsnorm{\hat{h}_{\jj\kk}(\beta)}^2$. Select a uniformly random $\ell_u \in [R_{\tj}]$ such that $\dim(\tau_{\ell_u})\geq 2$ and assign the label $\ell_u$ to $u$. If there is no such index, then assign $\bot$.

Now fix a constraint $\pi_{vu}$. The probability, $p_{vu}$, that this edge is satisfied by the randomized labeling, is lower bounded by:
\begin{align*}
p_{vu} &\geq \sum_{\substack{\alpha , \beta, \\ \dim(\alpha) , \dim(\beta)\geq 2, \\  \dimgeqi{\alpha}{2} \leq C, \\ \beta_{\geq 2} \subseteq \pi_{\geq 2}(\alpha)}}\dim(\alpha)\dim(\beta) \hsnorm{\hat{g'}(\alpha)}^2  \hsnorm{\hat{h}(\beta)}^2 \cdot  \frac{1}{\dimgeqi{\alpha}{2}}\\
&\geq \frac{1}{C}\sum_{\substack{\alpha , \beta, \\ \dim(\alpha) , \dim(\beta)\geq 2, \\  \dimgeqi{\alpha}{2} \leq C, \\ \beta_{\geq 2} \subseteq \pi_{\geq 2}(\alpha)}}\dim(\alpha)\dim(\beta) \hsnorm{\hat{g'}(\alpha)}^2  \hsnorm{\hat{h}(\beta)}^2   \\
&\geq \frac{1}{C\cdot |G|^C}	\left|\Theta^{(v,u)}_{\ii,\jj,\kk}(\low)\right|^2 .
\end{align*}
Therefore, the expected fraction of constraints between the layers $\ti$ and $\tj$ satisfied by the randomized labeling is lower bounded by 
\begin{align*}
\E_{\pi_{vu}\in \Pi_{\ti, \tj}} [p_{vu}] &\geq \E_{\pi_{vu}\in \Pi_{\ti, \tj}}\left[\frac{1}{C\cdot |G|^C}	|\Theta^{(v,u)}_{\ii,\jj,\kk}(\low)|^2\right] \\
& = \frac{1}{C\cdot |G|^C}	 \E_{\pi_{vu}\in \Pi_{\ti, \tj}}\left[|\Theta^{(v,u)}_{\ii,\jj,\kk}(\low)|^2\right]\\
& \geq \frac{1}{C\cdot |G|^C}	 \E_{\pi_{vu}\in \Pi_{\ti, \tj}}\left[|\Theta^{(v,u)}_{\ii,\jj,\kk}(\low)|\right]^2\tag*{(Using convexity)}\\
  & \geq \frac{1}{C\cdot |G|^C}	\cdot  \frac{\delta^2}{4|G|^6}\\
  &> \frac{\delta^2}{10|G|^{10C}}.
\end{align*}
Since the expected fraction of the edges that are satisfied is strictly greater than $ \frac{\delta^2}{10|G|^{10C}}$, by conditional expectation, there exists a labeling to the Label Cover instance $\calH$ that satisfies more than $ \frac{\delta^2}{10|G|^{10C}}$ fraction of the edges between the layers $\ti$ and $\tj$, which is a contradiction.
\end{proof}

\subsubsection{Bounding the \texorpdfstring{$\high$}{high} terms}
\label{sec:high-terms}

We now show the following claim:

\begin{claim}(Restatement of Claim~\ref{claim:highterms}) Suppose $\calH$ is at most $\frac{\delta^{10}}{(2|G|)^{100}}$-satisfiable. Then for every $\rho\in \irr(G)$ such that $\dim(\rho)\geq 2$ and $1\leq\ii,\jj,\kk\leq \dim(\rho)$, 
	$$\left|\E_{1\leq \ti<\tj\leq \ell}~~\E_{\pi_{vu}\in \Pi_{\ti,\tj}}\left[\Theta^{(v,u)}_{\ii,\jj,\kk} (\high)\right] \right| \leq \frac{\delta}{2|G|^3}.$$
\end{claim}
\begin{proof}
We begin by giving an alternate but equivalent way of sampling the constraint $\pi_{vu}$ in the expectation from the claim. Sample a random path $\mathrm{p} = (v_1, v_2, \ldots, v_\ell)$ as follows. Sample $v_1\in V_1$ u.a.r. and then for each $2\leq i\leq \ell$, sample a uniformly random neighbor $v_i$ of $v_{i-1}$ from the graph $G_{i-1, i}$. Next, sample a uniformly random pair $(\ti, \tj)$ such that $1\leq \ti<\tj\leq \ell$ and output the constraint $\pi_{v_{\ti}v_{\tj}}$. Note that by the property (2) from Theorem~\ref{theorem:multiLC}, the above distribution is identical to the distribution of first sampling a random pair $(\ti, \tj)$ and then sampling a random constraint from the graph $G_{\ti, \tj}$. Our goal is to show the following bound
$$\Delta:= \left|\E_{\mathrm{p}=(v_1, v_2, \ldots, v_\ell)} \E_{1\leq \ti<\tj\leq \ell}~~ \left[\Theta^{(v_{\ti},v_{\tj})}_{\ii,\jj,\kk} (\high)\right]\right| \leq \frac{\delta}{2|G|^3}.$$

In what follows, we fix a path $\mathrm{p} = (v_1, v_2, \ldots, v_\ell)$ and the layer $\ti$ and simplify the expression inside the expectation for this path and the layer $\ti$. Recall,
$$\Theta^{(v_{\ti},v_{\tj})}_{\ii,\jj,\kk} (\high): = \sum_{\substack{\alpha,\beta , \\ \dim(\alpha) , \dim(\beta)\geq 2, \\ \dimgeqi{\alpha}{2} >  C}} {\termm{(v_{\ti},v_{\tj})}{\alpha}{\beta}}.$$

 For notational convenience, we suppress the conditions on $\alpha,\beta$ and simply write the sum over pairs $\alpha, \beta$. We will analyze the complete sum with the extra conditions on $\alpha, \beta$, once we simplify the expression. 

Let $U(\mathrm{p}, \alpha, t)$ be a transformation (i.e., change of basis) which takes a representation $ \alpha(.)$ and converts it into a direct sum of irreducible representations of $\pi_{v_{\ti}v_{\ell}}(G^{R_{\ti}}) := \{\V x\circ \pi_{v_{\ti}v_{\ell}} \mid \V x\in G^{R_{\ell}}\}$ which is a subgroup of $G^{R_{\ti}}$ isomorphic to the group $G^{R_{\ell}}$. In other words, if we let $U= U(\mathrm{p}, \alpha, t)$, then $U\alpha(\V x\circ \pi_{v_{\ti}v_{\ell}} )U^\star$ is a block diagonal matrix for each $\V x\in G^{R_{\ell}}$ with the same block structure. 

Although  the irreducible
isomorphism classes and their multiplicities in $ \alpha(.)$ are uniquely determined, the choice of $U(\mathrm{p}, \alpha, t)$ is not unique. For our purposes, we choose $U(\mathrm{p}, \alpha, t)$ with the following additional structure which will simplify things towards the end of the analysis. We choose $U(\mathrm{p}, \alpha, t)$ such that the following holds for every $\ti <j < \ell$: The block matrices  $U\alpha(\V x\circ \pi_{v_{\ti}v_{j+1}} )U^\star$ corresponding to $\V x\in G^{R_{j+1}}$ are a {\em refinement} of the block matrices $U\alpha(\V x\circ \pi_{v_{\ti}v_{j}} )U^\star$ corresponding to $\V x\in G^{R_{j}}$, as depicted in Figure~\ref{fig:refinement-block}. Note that this is possible from Proposition~\ref{prop:refinement} as $\{\V x\circ \pi_{v_{\ti}v_{j+1}}  \mid \V x\in G^{R_{j+1}}\}$ is a subgroup of $\{ \V x\circ \pi_{v_{\ti}v_{j}}\mid \V x\in G^{R_{j}}\}$.

\begin{figure}
    \centering
\begin{tikzpicture}[x=0.65cm,y=0.65cm,>=Stealth]

%==================== LEFT: coarse (3 diagonal blocks) ====================
\begin{scope}[shift={(0,0)}]
  % outer box
  \draw[thick] (0,0) rectangle (8,8);

  % coarse partition lines (3 blocks with sizes 3,2,3)
  \draw[thick] (3,5)--(3,8);
  \draw[thick] (3,5)--(0,5);

  \draw[thick] (5,3)--(5,5);
  \draw[thick] (3,5)--(5,5);
  \draw[thick] (3,3)--(5,3);
  \draw[thick] (3,5)--(3,3);
  
  \draw[thick] (5,3)--(8,3);
  \draw[thick] (5,3)--(5,0);
  % shade the 3 diagonal blocks
  \fill[pattern=north east lines,pattern color=black!70] (0,5) rectangle (3,8);
  \fill[pattern=north east lines,pattern color=black!70] (3,3) rectangle (5,5);
  \fill[pattern=north east lines,pattern color=black!70] (5,0) rectangle (8,3);
\end{scope}

% arrow + label
\draw[->,thick] (9,4) -- (11,4) node[midway,above]{refinement};

%==================== RIGHT: refined ====================
\begin{scope}[shift={(12,0)}]
  % outer box
  \draw[thick] (0,0) rectangle (8,8);

  % show the original coarse diagonal blocks (dashed)
  \draw[dashed,thick] (0,5) rectangle (3,8);
  \draw[dashed,thick] (3,3) rectangle (5,5);
  \draw[dashed,thick] (5,0) rectangle (8,3);

  % ---- Refinement inside block 1 (size 3 -> 1 + 2) ----
  % small 1x1 at top-left corner of the first coarse block
  \fill[pattern=north east lines,pattern color=blue] (0,7) rectangle (1,8);
  \draw[thick] (0,7) rectangle (1,8);

  % 2x2 block (diagonal) inside the first coarse block
  \fill[pattern=north east lines,pattern color=blue] (1,5) rectangle (3,7);
  \draw[thick] (1,5) rectangle (3,7);

  % ---- Refinement inside block 2 (size 2 -> 0.6 + 0.6 + 0.8) ----
  \fill[pattern=north east lines,pattern color=blue] (3.0,4.4) rectangle (3.6,5.0);
  \draw[thick] (3.0,4.4) rectangle (3.6,5.0);

  \fill[pattern=north east lines,pattern color=blue] (3.6,3.8) rectangle (4.2,4.4);
  \draw[thick] (3.6,3.8) rectangle (4.2,4.4);

  \fill[pattern=north east lines,pattern color=blue] (4.2,3.0) rectangle (5.0,3.8);
  \draw[thick] (4.2,3.0) rectangle (5.0,3.8);

  % ---- Refinement inside block 3 (size 3 -> 0.6 + 1.8 + 0.6) ----
  \fill[pattern=north east lines,pattern color=blue] (5.0,2.4) rectangle (5.6,3.0);
  \draw[thick] (5.0,2.4) rectangle (5.6,3.0);

  \fill[pattern=north east lines,pattern color=blue] (5.6,0.6) rectangle (7.4,2.4);
  \draw[thick] (5.6,0.6) rectangle (7.4,2.4);

  \fill[pattern=north east lines,pattern color=blue] (7.4,0.0) rectangle (8.0,0.6);
  \draw[thick] (7.4,0.0) rectangle (8.0,0.6);

\end{scope}

\end{tikzpicture}
    \caption{Refinement of a block diagonal matrix}
    \label{fig:refinement-block}
\end{figure}

For ease of notation, we denote this unitary matrix, which depends on $\mathrm{p}, \alpha, t$, by $U$. We have
\begin{align}
\trace(\hat{g}(\alpha)\hat{g'}(\alpha)  {\alpha}(\V b\circ \pi_{v_\ti v_{\tj}})) &= \trace(U\hat{g}(\alpha)\hat{g'}(\alpha)  {\alpha}(\V b\circ \pi_{v_\ti v_{\tj}})U^\star)\tag*{(cyclic property of $\trace$, and $UU^\star = I$)} \nonumber\\
&= \trace(U\hat{g}(\alpha)\hat{g'}(\alpha) U^\star U {\alpha}(\V b\circ \pi_{v_\ti v_{\tj}})U^\star).\label{eq:identity_trace}
\end{align}
In this last expression, $U {\alpha}(\V b\circ \pi_{v_\ti v_{\tj}})U^\star$ is a block diagonal matrix.

To simplify the notation further, we are going to suppress the use of $U$ and write:
$$U\hat{g}(\alpha) =: A(\alpha), \quad \quad \hat{g'}(\alpha) U^\star =: A'(\alpha) \quad \quad\mbox{ and } \quad \quad U {\alpha}(\V b\circ \pi_{v_\ti v_{\tj}})U^\star =: B_{\ti, \tj}(\alpha)(\V b).$$

Note that by Claim~\ref{claim:Uni_norm}, the $\hsnorm{\cdot}$ of the matrices are preserved, i.e., $\hsnorm{A(\alpha)} = \hsnorm{\hat{g}(\alpha)}$ and $\hsnorm{A'(\alpha)} = \hsnorm{\hat{g'}(\alpha)}$. Coming back to the task of simplifying the expression:

\begin{align*}
\Theta^{(v_{\ti},v_{\tj})}_{\ii,\jj,\kk} (\high) & = \sum_{\alpha , \beta} \dim(\alpha)\dim(\beta) \E_{\V b} \left[\trace(\hat{g}(\alpha)\hat{g'}(\alpha)  {\alpha}(\V b\circ \pi)) \cp  \trace(\hat{h}(\beta)  {\beta}(\V b^{-1}))\right]\\
& = \sum_{\alpha , \beta} \dim(\alpha)\dim(\beta) \E_{\V b} \left[\trace(U\hat{g}(\alpha)\hat{g'}(\alpha) U^\star  U {\alpha}(\V b\circ \pi)U^\star) \cp  \trace(\hat{h}(\beta)  {\beta}(\V b^{-1}))\right]\\
& = \sum_{\alpha , \beta} \dim(\alpha)\dim(\beta) \E_{\V b} \left[\trace(A(\alpha) A'(\alpha)  B_{\ti, \tj}(\alpha)(\V b)) \cp  \trace(\hat{h}(\beta)  {\beta}(\V b^{-1}))\right]\\
& = \sum_{\alpha , \beta} \dim(\alpha)\dim(\beta) \E_{\V b} \left[\sum_{i,k} (A(\alpha)A'(\alpha))_{ik}\cdot  (B_{\ti, \tj}(\alpha)(\V b))_{ki} \cp \sum_{i',k'}\hat{h}(\beta)_{i'k'}  {\beta}(\V b^{-1})_{k'i'}\right]\\
& = \sum_{\alpha , \beta} \dim(\alpha)\dim(\beta) \sum_{i,k} (A(\alpha)A'(\alpha))_{ik}\cdot  \sum_{i',k'} \hat{h}(\beta)_{i'k'} \E_{\V b}  \left[ (B_{\ti, \tj}(\alpha)(\V b))_{ki} \cp   {\beta}(\V b^{-1})_{k'i'}\right]\\
\end{align*}

{\color{black}
%Fix the path, $\ti<\tj$, and $\alpha$, and retain the same unitary
%$U=U(\mathrm p,\alpha,\ti)$ for all later layers.  

 Let $\beta_{U, \alpha, \tj, i}$ denote the irreducible representation of $G^{R_{\tj}}$ present in the $i^{th}$ column/row of the block matrix $B_{\ti, \tj}(\alpha)(\cdot)$. 
Index the individual
irreducible blocks of $B_{\ti,\tj}(\alpha)$ by $q$, and write $I_q$ for
the coordinate set of block $q$.  For each block choose its representative
$\beta_q\in\irr(G^{R_{\tj}})$ and a unitary intertwiner $S_q$ so that
\[
 B_{\ti,\tj}(\alpha)(\V b)|_{I_q}
   =S_q\beta_q(\V b)S_q^\star.
\]
The intertwiner may depend on the block and on $\tj$; it does not
change $U$.  Define
$K_q:=S_q\hat h(\beta_q)S_q^\star$.  

The matrix-coefficient orthogonality in
Proposition~\ref{prop:orthmatrixentries}, applied in the $\beta_q$ basis,
now yields
\[
 \dim(\beta)\E_{\V b}\!\left[
 (B_{\ti,\tj}(\alpha)(\V b))_{ki}
 \trace\bigl(\hat h(\beta)\beta(\V b^{-1})\bigr)\right]
 =\begin{cases}
 (K_q)_{ki}, & i,k\in I_q\text{ and }\beta=\beta_q,\\
 0, & \text{otherwise}.
 \end{cases}
\]
In particular, a block entry may pair with several matrix entries in
the chosen $\beta_q$ basis; the conjugated coefficient $K_q$ accounts
for all of them.  If $q(k)$ denotes the unique block containing $k$,
\[
 \Theta^{(v_{\ti},v_{\tj})}_{\ii,\jj,\kk}(\high)
 =\sum_{\substack{\alpha\\\dim_{\geq2}(\alpha)>C}}
   \dim(\alpha)\sum_{j,k} A'(\alpha)_{jk}
      \sum_{i\in I_{q(k)}} A(\alpha)_{ij}(K_{q(k)})_{ki}.
\]
By Parseval and the unitary invariance of the Hilbert--Schmidt norm,
$\sum_\alpha\dim(\alpha)\|A'(\alpha)\|_{\mathrm{HS}}^2
\leq\|g'\|_2^2\leq1$.  Apply Cauchy--Schwarz to the last display,
and for each block put
$a_{qj}:=(A(\alpha)_{ij})_{i\in I_q}$.  Then
\begin{align*}
 |\Theta^{(v_{\ti},v_{\tj})}_{\ii,\jj,\kk}(\high)|^2
 &\leq \sum_{\substack{\alpha\\\dim_{\geq2}(\alpha)>C}}
    \dim(\alpha)\sum_{j,q}\sum_{k\in I_q}
       \left|\sum_{i\in I_q} A(\alpha)_{ij}(K_q)_{ki}\right|^2\\
 &=\sum_{\substack{\alpha\\\dim_{\geq2}(\alpha)>C}}
    \dim(\alpha)\sum_{j,q}\|K_q a_{qj}\|_2^2\\
 &\leq\sum_{\substack{\alpha\\\dim_{\geq2}(\alpha)>C}}
    \dim(\alpha)\sum_{j,q}\|K_q\|_{\mathrm{HS}}^2
       \sum_{i\in I_q}|A(\alpha)_{ij}|^2.
\end{align*}
Since, using Claim~\ref{claim:Uni_norm}, $\|K_q\|_{\mathrm{HS}}=\|\hat h(\beta_q)\|_{\mathrm{HS}}$
and $\beta_q=\beta_{U,\alpha,\tj,i}$ for $i\in I_q$,
we get
}

\begin{align}
|\Theta^{(v_{\ti},v_{\tj})}_{\ii,\jj,\kk} (\high)|^2 
& \leq \sum_{\substack{\alpha\\\dim_{\geq2}(\alpha)>C} } \dim(\alpha)  \sum_{\substack{i, j\\ }}  |(U\hat{g}(\alpha))_{ij}|^2 \cdot  \hsnorm{\hat{h}(\beta_{ U, \alpha, \tj, i})}^2.\label{eq:sim_high_exp}
\end{align}

We now come back to the expectation that we need to bound. In the following analysis, we split the expectation on $\ti, \tj$ and write those separately (they are still sampled as a uniformly random pair $1\leq \ti<\tj\leq \ell$). Applying the Cauchy-Schwarz inequality and using the upper bound from ~(\ref{eq:sim_high_exp}) above, we get

\begin{align*}
\Delta^2 &= \left|\E_{\mathrm{p}=(v_1, v_2, \ldots, v_\ell)} \E_{1\leq \ti<\tj\leq \ell}~~ \left[\Theta^{(v_{\ti},v_{\tj})}_{\ii,\jj,\kk} (\high)\right]\right|^2 \\
&~~~~~~~~~~~~~~~\leq     
\E_{\mathrm{p}=(v_1, v_2, \ldots, v_\ell)} \E_{\ti} \E_{\tj}~~ \left[\left|\Theta^{(v_{\ti},v_{\tj})}_{\ii,\jj,\kk} (\high)\right|^2\right]\\
&~~~~~~~~~~~~~~~ \leq \E_{\mathrm{p}=(v_1, v_2, \ldots, v_\ell)} \E_{\ti} \E_{\tj}~~ \left[\sum_{\alpha } \dim(\alpha)  \sum_{\substack{i, j\\ }}  |(U\hat{g}(\alpha))_{ij}|^2 \cdot  \hsnorm{\hat{h}(\beta_{ U, \alpha, \tj, i})}^2\right]\\
&~~~~~~~~~~~~~~~ = \E_{\mathrm{p}=(v_1, v_2, \ldots, v_\ell)} \E_{\ti} \E_{\tj}~~ \left[\sum_{\alpha } \dim(\alpha)  \sum_{\substack{i\\ }}  \|(U\hat{g}(\alpha))_{i}\|_2^2 \cdot  \hsnorm{\hat{h}(\beta_{ U, \alpha, \tj, i})}^2\right],
\end{align*}
where we denote the $\ell_2$ norm of the $i^{th}$ row of the matrix $U\hat{g}(\alpha)$ by $\|(U\hat{g}(\alpha))_{i}\|_2$.

By taking a closer look at the expression above, it seems that there can be multiple scenarios when the expression is large. For instance, it might happen that some $\beta_{ U, \alpha, \tj, i}$'s for which $\hsnorm{\hat{h}(\beta_{ U, \alpha, \tj, i})}^2$ is large can have small dimension, and in this case, we will not be able to get the advantage that we saw in the dictatorship test. Surprisingly, we will show next that if the expression is large because of $\beta_{ U, \alpha, \tj, i}$'s that have small dimensions, then this will give a strategy to assign labels to the \LC\ instance that will satisfy a non-trivial fraction of the constraints. 

We first bound the expression when the $\beta$s have a large dimension using the trick we saw in the dictatorship test. Towards this, we recall $C = \frac{8|G|^6}{\delta^2}$.

\begin{align*}
\Delta^2&\leq \E_{\mathrm{p}=(v_1, v_2, \ldots, v_\ell)} \E_{\ti} \E_{\tj}~~ \left[\sum_{\alpha } \dim(\alpha)  \sum_{\substack{i\\ }}  \|(U\hat{g}(\alpha))_{i}\|_2^2 \cdot  \hsnorm{\hat{h}(\beta_{ U, \alpha, \tj, i})}^2\right]\\
& = \E_{\mathrm{p}=(v_1, v_2, \ldots, v_\ell)} \E_{\ti} \E_{\tj}~~ \left[\sum_{\alpha } \dim(\alpha) \left( \sum_{\substack{i\mid \\\dim(\beta_{ U, \alpha, \tj, i}) \leq C }}  \|(U\hat{g}(\alpha))_{i}\|_2^2 \cdot  \hsnorm{\hat{h}(\beta_{ U, \alpha, \tj, i})}^2 \right)\right]  + \\
&~~~~~~~~~~~~\E_{\mathrm{p}=(v_1, v_2, \ldots, v_\ell)} \E_{\ti} \E_{\tj}~~ \left[  \sum_{\alpha } \dim(\alpha)\left(\sum_{\substack{i \mid\\\dim(\beta_{ U, \alpha, \tj, i}) > C }}\|(U\hat{g}(\alpha))_{i}\|_2^2 \cdot  \hsnorm{\hat{h}(\beta_{ U, \alpha, \tj, i})}^2\right)\right]\\
&= \Delta_1 + \Delta_2,
\end{align*}
where $\Delta_1$ is the former expectation and $\Delta_2$ is the latter expectation.
$\Delta_2$ can be upper-bounded as follows.
\begin{align*}
\Delta_2 &= \E_{\mathrm{p}=(v_1, v_2, \ldots, v_\ell)} \E_{\ti} \E_{\tj}~~ \left[  \sum_{\alpha } \dim(\alpha)\left(\sum_{\substack{i \mid\\\dim(\beta_{ U, \alpha, \tj, i}) > C }}\|(U\hat{g}(\alpha))_{i}\|_2^2 \cdot  \hsnorm{\hat{h}(\beta_{ U, \alpha, \tj, i})}^2\right)\right]\\
&\leq \E_{\mathrm{p}=(v_1, v_2, \ldots, v_\ell)} \E_{\ti} \E_{\tj}~~ \left[  \sum_{\alpha } \dim(\alpha)\hsnorm{\hat{g}(\alpha)}^2 \left(\sum_{\substack{\beta \mid\\ \exists i~ \text{s.t.}~ \beta = \beta_{ U, \alpha, \tj, i}\\\dim(\beta) > C }}  \hsnorm{\hat{h}(\beta)}^2\right)\right]\\
&\leq \frac{1}{C}\E_{\mathrm{p}=(v_1, v_2, \ldots, v_\ell)} \E_{\ti} \E_{\tj}~~ \left[  \sum_{\alpha } \dim(\alpha)\hsnorm{\hat{g}(\alpha)}^2 \left(\sum_{\substack{\beta\\ }} \dim(\beta)  \hsnorm{\hat{h}(\beta)}^2\right)\right]\\
&\leq \frac{1}{C} \E_{\mathrm{p}=(v_1, v_2, \ldots, v_\ell)} \E_{\ti} \E_{\tj}~~ \left[\|g\|_2^2 \|h\|_2^2\right] \leq \frac{1}{C}.
\end{align*}
Therefore, we have,
\begin{align*}
\Delta^2&\leq \Delta_1 + \frac{1}{C}.
\end{align*}

Thus, we are left with bounding $\Delta_1$. In what follows, we make the function $h$ explicit by writing its subscript. Recall, $h_v(\V y):=\rho(f_v(\V y))_{\jj\kk}$ for a vertex $v\in \calV$. We have,
\begin{align}
\Delta_1 = \E_{\mathrm{p}=(v_1, v_2, \ldots, v_\ell)} \E_{\ti} \E_{\tj}~~ \left[\sum_{\substack{\alpha~|\\ \dim_{\geq 2}(\alpha) > C}} \dim(\alpha) \left( \sum_{\substack{i\mid \\\dim(\beta_{ U, \alpha, \tj, i}) \leq C }}  \|(U\hat{g}(\alpha))_{i}\|_2^2 \cdot  \hsnorm{\hat{h}_{v_{\tj}}(\beta_{ U, \alpha, \tj, i})}^2 \right)\right].\label{eq:high_exp_1}
\end{align}

In order to intuitively make sense of the above expression, here is an alternate way to view the  above upper bound. Let $\alpha_V$ be the complex vector space on which $\alpha$ acts. Considering $\alpha(\cdot)$ as a representation of $G^{R_{\tj}}$, let $\beta_V$ be the span of invariant subspaces of $\alpha_V$, corresponding to (all copies of) the irreducible representation $\beta$ of $G^{R_{\tj}}$. For a vector $w\in \alpha_V$, denote the projection of $w$ onto the subspace $\beta_V$ by $P_{\beta_V}(w)$. The following is an alternate description of the quantity $\Delta_1$. 
\begin{align}
\Delta_1 &=   \E_{\mathrm{p}}\E_{\ti} \E_{\tj}~~ \left[\sum_{\substack{\beta~\mid\\ \dim(\beta)\leq C}} \hsnorm{\hat{h}_{v_{\tj}}(\beta)}^2 \sum_{\substack{\alpha~|\\ \dim_{\geq 2}(\alpha) > C}} \dim(\alpha) \sum_{1\leq j \leq \dim(\alpha)} \|P_{\beta_V}( \text{$j^{th}$ column of $\hat{g}(\alpha)$})\|_2^2\right].\label{eq:high_exp_2}
\end{align}
To see why this is the same expression as above, fix $\mathrm{p}$, $\ti$ and $\tj$. For simplicity, assume that $U = I$ and fix a representation $\beta\in \irr(G^{R_{\tj}})$ and a representation $\alpha\in \irr(G^{R_{\ti}})$. Now in the expression (\ref{eq:high_exp_1}), if we let $\beta_V =\mathrm{span}(\{e_{\ell_1}, e_{\ell_2}, \ldots, e_{\ell_m}\})$ be the invariant space of $\alpha_V$ corresponding to (all copies of) the representation $\beta$, then $\hsnorm{\hat{h}_{v_{\tj}}(\beta)}^2$ gets multiplied by the following terms corresponding to $\alpha$ with $\dim_{\geq 2}(\alpha) > C$
\begin{align*}
    \dim(\alpha)\sum_{k\in [m]} \|\hat{g}(\alpha)_{\ell_k}\|_2^2 &= \dim(\alpha)\sum_{k\in [m]}  \sum_{1\leq j\leq \dim(\alpha)}|\hat{g}(\alpha)_{\ell_k, j}|^2 \\
    &= \dim(\alpha)\sum_{1\leq j\leq \dim(\alpha)}\sum_{k\in [m]}  |\hat{g}(\alpha)_{\ell_k, j}|^2\\
&= \dim(\alpha)\sum_{1\leq j\leq \dim(\alpha)} \|P_{\beta_V}( \text{$j^{th}$ column of $\hat{g}(\alpha)$})\|_2^2,\end{align*}
as required.

Looking at the expressions (\ref{eq:high_exp_1})  and (\ref{eq:high_exp_2}), we make a few observations.
\begin{enumerate}
    \item  The quantity inside the expectations is upper-bounded by $\|h_{v_{\tj}}\|_2^2 \|g\|_2^2\leq 1$, 
        \item Looking at  (\ref{eq:high_exp_1}), for a fixed $\alpha$, the invariant spaces $\beta_V$ are mutually orthogonal for different $\beta$'s. Furthermore, $\beta$'s for which $\hsnorm{\hat{h}_{v_{\tj}}(\beta)}^2$ are small do not contribute much to $\Delta_1$ (see (\ref{eq:samll_hs_beta}) below). Now, suppose $\Delta_1$ is large, say at least $\eta>0$. It then follows from  (\ref{eq:high_exp_2}) that for at most $O_\eta(1)$ $\beta$'s, $\hsnorm{\hat{h}_{v_{\tj}}(\beta)}^2\geq \Omega_\eta(1)$, and the contribution from the second sum is $\Omega_\eta(1)$ for these $\beta$'s; Call these $\beta$s {\em heavy}. 
        \item Finally, if $\Delta_1$ is large, then for a constant fraction of paths $\mathrm{p} = (v_1, v_2, \ldots, v_\ell)$ and $t\in [\ell]$, for constant fraction of pairs $t< t_1< t_2\leq \ell$, the span of invariant spaces for the heavy $\beta$'s for $v_{t_1}$ must not be mutually orthogonal to the span of invariant spaces for the heavy $\beta$'s for $v_{t_2}$. Otherwise, in the second sum in (\ref{eq:high_exp_2}), we get projections of columns of $\hat{g}(\alpha)$ on mutually orthogonal subspaces as we vary $\tj$, and thus the expectation over $\tj$ would make $\Delta_1$ small. As the dimension of $\beta$'s under consideration is small, this allows us to deduce consistent labeling for a constant fraction of pairs of layers. 
\end{enumerate}

We make the above intuition formal in the following claim.

\begin{claim}
\label{claim:labeling-from-high}
Let $\eta = \frac{\delta^2}{8|G|^6}$. If $\calH$ is less than $ \left( \frac{\eta}{2\log_2 C}\right)^2\cdot \frac{\eta^2}{64}$-satisfiable,  then $\Delta_1 \leq \eta$.
\end{claim}
\begin{proof}
    Suppose $\Delta_1\geq \eta$. From (\ref{eq:high_exp_1}), we have
  \begin{align*}
      \eta&\leq \E_{\mathrm{p}=(v_1, v_2, \ldots, v_\ell)} \E_{\ti} \E_{\tj}~~ \left[\sum_{\substack{\alpha~|\\ \dim_{\geq 2}(\alpha) > C}} \dim(\alpha) \left( \sum_{\substack{i\mid \\\dim(\beta_{ U, \alpha, \tj, i}) \leq C }}  \|(U\hat{g}(\alpha))_{i}\|_2^2 \cdot  \hsnorm{\hat{h}_{v_{\tj}}(\beta_{ U, \alpha, \tj, i})}^2 \right)\right]\\
    &= \E_{\mathrm{p}=(v_1, v_2, \ldots, v_\ell)} \E_{\ti} \E_{\tj}~~ \left[\sum_{\substack{\alpha~|\\ \dim_{\geq 2}(\alpha) > C} } \dim(\alpha) \left( \sum_{\substack{i \\ }}  \|(U\hat{g}(\alpha))_{i}\|_2^2 \cdot  \hsnorm{\hat{h}_{v_{\tj}}(\beta_{ U, \alpha, \tj, i})}^2 \cdot \mathbf{1}_{\dim(\beta_{ U, \alpha, \tj, i}) \leq C} \right)\right]\\
     &\leq  \E_{\mathrm{p}=(v_1, v_2, \ldots, v_\ell)} \E_{\ti} \E_{\tj}~~ \left[\sum_{\substack{\alpha~|\\ \dim_{\geq 2}(\alpha) > C} } \dim(\alpha) \left( \sum_{\substack{i \\ }}  \|(U\hat{g}(\alpha))_{i}\|_2^2 \cdot  \hsnorm{\hat{h}_{v_{\tj}}(\beta_{ U, \alpha, \tj, i})}^2 \cdot \mathbf{1}_{\substack{\dim(\beta_{ U, \alpha, \tj, i}) \leq C ~\wedge~\\  \hsnorm{\hat{h}_{v_{\tj}}(\beta_{ U, \alpha, \tj, i})}^2\geq \eta/2}} \right)\right] + \frac{\eta}{2}.\\
     \end{align*}
     The last step follows from the following bound
     \begin{align}
         &\sum_{\substack{\alpha~|\\ \dim_{\geq 2}(\alpha) > C} } \dim(\alpha) \left( \sum_{\substack{i \\ }}  \|(U\hat{g}(\alpha))_{i}\|_2^2 \cdot  \hsnorm{\hat{h}_{v_{\tj}}(\beta_{ U, \alpha, \tj, i})}^2 \cdot \mathbf{1}_{\dim(\beta_{ U, \alpha, \tj, i}) \leq C ~\wedge~ \hsnorm{\hat{h}_{v_{\tj}}(\beta_{ U, \alpha, \tj, i})}^2< \eta/2} \right)\nonumber\\
         &~~~~~~~~~~~~~~\leq \sum_{\substack{\alpha~|\\ \dim_{\geq 2}(\alpha) > C}} \dim(\alpha) \left( \sum_{\substack{i \\ }}  \|(U\hat{g}(\alpha))_{i}\|_2^2 \cdot  \frac{\eta}{2} \right)\nonumber\\
          &~~~~~~~~~~~~~~=\frac{\eta}{2} \sum_{\substack{\alpha~|\\ \dim_{\geq 2}(\alpha) > C}} \dim(\alpha) \left( \sum_{\substack{i \\ }}  \|(U\hat{g}(\alpha))_{i}\|_2^2\right)\nonumber\\
        &~~~~~~~~~~~~~~=\frac{\eta}{2} \sum_{\substack{\alpha~|\\ \dim_{\geq 2}(\alpha) > C} } \dim(\alpha) \hsnorm{U\hat{g}(\alpha)}^2 \nonumber\\
            &~~~~~~~~~~~~~~\leq \frac{\eta}{2}\|g\|_2^2\leq \frac{\eta}{2}.\label{eq:samll_hs_beta}
     \end{align}
     Therefore, we have
     \begin{align*}
     \frac{\eta}{2} &\leq  \E_{\mathrm{p}=(v_1, v_2, \ldots, v_\ell)} \E_{\ti} \E_{\tj}~~ \left[\sum_{\alpha } \dim(\alpha) \left( \sum_{\substack{i \\ }}  \|(U\hat{g}(\alpha))_{i}\|_2^2 \cdot  \hsnorm{\hat{h}_{v_{\tj}}(\beta_{ U, \alpha, \tj, i})}^2 \cdot \mathbf{1}_{\dim(\beta_{ U, \alpha, \tj, i}) \leq C ~\wedge~ \hsnorm{\hat{h}_{v_{\tj}}(\beta_{ U, \alpha, \tj, i})}^2\geq \eta/2} \right)\right] \\
     &\leq  \E_{\mathrm{p}=(v_1, v_2, \ldots, v_\ell)} \E_{\ti} ~~ \left[\max_{\alpha, i} \left\{\E_{\tj}\left[ \hsnorm{\hat{h}_{v_{\tj}}(\beta_{ U, \alpha, \tj, i})}^2 \cdot \mathbf{1}_{\dim(\beta_{ U, \alpha, \tj, i}) \leq C ~\wedge~ \hsnorm{\hat{h}_{v_{\tj}}(\beta_{ U, \alpha, \tj, i})}^2\geq \eta/2} \right]\right\} \|g\|_2^2\right]\\
     &\leq  \E_{\mathrm{p}=(v_1, v_2, \ldots, v_\ell)} \E_{\ti} ~~ \left[\max_{\alpha, i} \left\{\E_{\tj}\left[ \hsnorm{\hat{h}_{v_{\tj}}(\beta_{ U, \alpha, \tj, i})}^2 \cdot \mathbf{1}_{\dim(\beta_{ U, \alpha, \tj, i}) \leq C ~\wedge~ \hsnorm{\hat{h}_{v_{\tj}}(\beta_{ U, \alpha, \tj, i})}^2\geq \eta/2} \right]\right\} \right]
    \end{align*}
    To reiterate, in the last expression, $\alpha$ varies over all the irreducible representations of $G^{R_{t}}$ and $1\leq i\leq \dim(\alpha)$. As the quantity inside the expectation is at most $1$, the above implies that there exists a small enough $\ti^\star\leq \ell-\sqrt{\ell}$ such that,
    \begin{align*}
     \frac{\eta}{4} \leq \frac{\eta}{2} - \frac{1}{\sqrt{\ell}} 
    &\leq  \E_{\mathrm{p}=(v_1, v_2, \ldots, v_\ell)} ~~ \left[\max_{\alpha, i} \left\{\E_{\tj>\ti^\star}\left[ \hsnorm{\hat{h}_{v_{\tj}}(\beta_{ U, \alpha, \tj, i})}^2 \cdot \mathbf{1}_{\dim(\beta_{ U, \alpha, \tj, i}) \leq C ~\wedge~ \hsnorm{\hat{h}_{v_{\tj}}(\beta_{ U, \alpha, \tj, i})}^2\geq \eta/2} \right]\right\}\right]\\
     &\leq  \E_{\mathrm{p}=(v_1, v_2, \ldots, v_\ell)} ~~ \left[\max_{\alpha, i} \left\{\E_{\tj>\ti^\star}\left[  \mathbf{1}_{\dim(\beta_{ U, \alpha, \tj, i}) \leq C ~\wedge~ \hsnorm{\hat{h}_{v_{\tj}}(\beta_{ U, \alpha, \tj, i})}^2\geq \eta/2} \right]\right\}\right],
    \end{align*}
    where we used the fact that  $\sqrt{\ell} \geq \frac{8}{\eta}$. In the above expressions, $\alpha$ refers to a representation of $G^{R_{t^\star}}$ and $\beta_{ U, \alpha, \tj, i}$ is the representation in the $i^{th}$ row of the block diagonal matrix $U(\mathrm{p}, \alpha, t^\star)\alpha(\cdot)U^\star(\mathrm{p}, \alpha, t^\star)$. Now suppose for a path $\mathrm{p}=(v_1, v_2, \ldots, v_\ell)$ for all $\alpha$ and $i$ there exists at most one $\tj>\ti^\star$ such that the condition $$\mathbf{1}_{\dim(\beta_{ U, \alpha, \tj, i}) \leq C ~\wedge~ \hsnorm{\hat{h}_{v_{\tj}}(\beta_{ U, \alpha, \tj, i})}^2\geq \eta/2}$$ holds. In this case, the quantity inside the expectation is at most $\frac{1}{\ell-\ti^\star}\leq \frac{1}{\sqrt{\ell}}$. Using this, and the fact that  $\sqrt{\ell} \geq \frac{8}{\eta}$, we get
     \begin{align*}
   \frac{\eta}{8}\leq \frac{\eta}{4} -\frac{1}{\sqrt{\ell}} &\leq  \E_{\mathrm{p}=(v_1, v_2, \ldots, v_\ell)} ~~ \left[\substack{\exists (\alpha, i), \exists t_2>t_1>\ti^\star ~\text{s.t.}~\\ \mathbf{1}_{\dim(\beta_{ U, \alpha, t_1, i}) \leq C ~\wedge~ \hsnorm{\hat{h}_{v_{t_1}}(\beta_{ U, \alpha, t_1, i})}^2\geq \eta/2}\cdot \mathbf{1}_{\dim(\beta_{ U, \alpha, t_2, i}) \leq C ~\wedge~ \hsnorm{\hat{h}_{v_{t_2}}(\beta_{ U, \alpha, t_2, i})}^2\geq \eta/2}} \right]
    \end{align*}
    By Jensen's inequality and averaging over pairs, there exists  $t^\star_2>t^\star_1>\ti^\star$ such that 
     \begin{align*}
 \frac{\eta^2}{64}&\leq  \E_{\mathrm{p}=(v_1, v_2, \ldots, v_\ell)} ~~ \left[\exists (\alpha, i),~ \mathbf{1}_{\dim(\beta_{ U, \alpha, t^\star_1, i}) \leq C ~\wedge~ \hsnorm{\hat{h}_{v_{t^\star_1}}(\beta_{ U, \alpha, t^\star_1, i})}^2\geq \eta/2}\cdot \mathbf{1}_{\dim(\beta_{ U, \alpha, t^\star_2, i}) \leq C ~\wedge~ \hsnorm{\hat{h}_{v_{t^\star_2}}(\beta_{ U, \alpha, t^\star_2, i})}^2\geq \eta/2} \right].
    \end{align*}

 Define $\operatorname{supp}_{\geq 2}(\beta):=\{q:\dim(\beta_q)\geq 2\}$. We now fix the above three layers  $t^\star_2>t^\star_1>\ti^\star$ and  consider the following two events when sampling a random path $\mathrm{p}=(v_1, v_2, \ldots, v_\ell)$:\\
    
\noindent{$\mathcal{E}_1$.} $\exists (\alpha, i)$, such that $\left({\dim(\beta_{ U, \alpha, t^\star_1, i}) \leq C ~\wedge~ \hsnorm{\hat{h}_{v_{t^\star_1}}(\beta_{ U, \alpha, t^\star_1, i})}^2\geq \eta/2}\right)$ and 

$\quad\quad\left({\dim(\beta_{ U, \alpha, t^\star_2, i}) \leq C ~\wedge~ \hsnorm{\hat{h}_{v_{t^\star_2}}(\beta_{ U, \alpha, t^\star_2, i})}^2\geq \eta/2}\right)$\\

\noindent{$\mathcal{E}_2$.} $\exists \beta_1, \beta_2~\text{such that}~\dim(\beta_1)\leq C, \dim(\beta_2)\leq C, \hsnorm{\hat{h}_{v_{t^\star_1}}(\beta_1)}^2 \geq \eta/2, ~\hsnorm{\hat{h}_{v_{t^\star_2}}(\beta_2)}^2 \geq \eta/2,$
$\pi_{v_{t^\star_1}v_{t^\star_2}}(\operatorname{supp}_{\geq 2}(\beta_1))\cap \operatorname{supp}_{\geq 2}(\beta_2)\neq \emptyset$\\

We now show that the event $\mathcal{E}_1$ implies the event $\mathcal{E}_2$. Suppose the event $\mathcal{E}_1$ holds with $(\alpha, i)$. Let $\beta_1 = \beta_{ U, \alpha, t^\star_1, i}$ and $\beta_2 = \beta_{ U, \alpha, t^\star_2, i}$ from the event $\mathcal{E}_1$  and furthermore $\beta_1 = (\rho_1, \rho_2, \ldots, \rho_{R_{t^\star_1}})$ and $\beta_2 = (\tau_1, \tau_2, \ldots, \tau_{R_{t^\star_2}})$.  We show that this $(\beta_1, \beta_2)$ satisfies the event $\mathcal{E}_2$.  For this, we need to look at the $i^{th}$ row of the block matrices $U\alpha(\V b\circ\pi_{v_{t^\star}v_{t_2^\star}})U^\star$ for all $\V b\in G^{R_{t^\star_2}}$ and $U\alpha(\V a\circ\pi_{v_{t^\star}v_{t_1^\star}})U^\star$ for all $\V a\in G^{R_{t^\star_1}}$. The $i^{th}$ row of the block matrices $U\alpha(\V a\circ\pi_{v_{t^\star}v_{t_1^\star}})U^\star$ corresponds to the representation $\beta_1$ and $i^{th}$ row of the block matrices $U\alpha(\V b\circ\pi_{v_{t^\star}v_{t_2^\star}})U^\star$ correspond to the representation $\beta_2$. In other words, if we write $\beta_1$ as a direct sum of representations of the subgroup $\pi_{v_{t^\star_1}, v_{t^\star_2}}(G^{R_{t^\star_1}})$, then $\beta_2$ appears as one of the representations in the direct sum $\beta_1$. For the blocks $\pi_{v_{t^\star_1}, v_{t^\star_2}}^{-1}(i)$ where $i\in [R_{t^\star_2}]$ such that $\otimes_{j\in{\pi_{v_{t^\star_1}, v_{t^\star_2}}^{-1}(i)}} \rho_j$ has dimension $1$, the corresponding $\tau_i$ will also have dimension $1$. It then follows that for every $i\in [R_{t^\star_2}]$ such that $\dim(\tau_i)\geq 2$, we have $i\in  \pi_{v_{t^\star_1}v_{t^\star_2}}(\operatorname{supp}_{\geq 2}(\beta_1))$ (one such $i$ exists as $\dim_{\geq 2}(\beta_2)\neq 0$ from Lemma~\ref{claim:folded_zero_coeff}). Thus $\mathcal{E}_1\implies \mathcal{E}_2$. Therefore, we have

$$ \E_{\mathrm{p}=(v_1, v_2, \ldots, v_\ell)} ~~ \left[ \mathbf{1}(\mathcal{E}_2)\right]\geq \frac{\eta^2}{64}.$$

 Define the following label set for a vertex $v\in V_{t^\star_1}\cup V_{t^\star_2}$
$$\mathcal{L}[v] =\bigcup \left\{ \operatorname{supp}_{\geq 2}(\beta) \mid \dim(\beta)\leq C~\text{and}~\hsnorm{\hat{h}_v(\beta)}^2\geq \frac{\eta}{2}\right\}.$$
As we know that
$$1\geq \|h_v\|_2^2 = \sum_{\beta} \dim(\beta) \hsnorm{\hat{h}_v(\beta)}^2,$$
we have $|\mathcal{L}[v]|\leq \frac{2\log_2C}{\eta}$ since each $\beta$ with $\dim(\beta)\leq C$ contributes at most $\log_2 C$ to the set $\mathcal{L}[v]$ and there can be at most $\frac{2}{\eta}$ $\beta$s such that $\hsnorm{\hat{h}_v(\beta)}^2 \geq \eta/2$.

Consider the following strategy to assign labels to vertices from layers $V_{t^\star_1}$ and $V_{t^\star_2}$. For each $v\in V_{t^\star_1}\cup V_{t^\star_2}$, assign a random label from $\mathcal{L}[v]$ to $v$ independently if $\mathcal{L}[v]$ is non-empty; otherwise, assign a random label from the respective set.

The expected fraction of the constraints in $\Pi_{t^\star_1, t^\star_2}$ satisfied by the above random assignment is lower bounded by the following expression. Note that when $\mathcal{E}_2$ happens then $1\leq |\mathcal{L}[{v_{t^\star_1}}]|, |\mathcal{L}[{v_{t^\star_2}}]|$; when $\mathbf{1}(\mathcal{E}_2)\equiv 0$, we use the convention that the quantity inside the expectation below is $0$.

 \begin{align*}
\E_{\mathrm{p}=(v_1, v_2, \ldots, v_\ell)} ~~ \left[\frac{1}{|\mathcal{L}[{v_{t^\star_1}}]||\mathcal{L}[{v_{t^\star_2}}]|} \mathbf{1}(\mathcal{E}_2) \right]\geq \left( \frac{\eta}{2\log_2 C}\right)^2\cdot \frac{\eta^2}{64}.
    \end{align*}
For our choice of parameters, $\frac{\delta^{10}}{(2|G|)^{100}} <\left( \frac{\eta}{2\log_2 C}\right)^2\cdot\frac{\eta^2}{64}$. Since $\calH$ is at most $\frac{\delta^{10}}{(2|G|)^{100}}$-satisfiable, we arrive at a contradiction and therefore  $\Delta_1\leq \frac{\delta^2}{8|G|^6}$ as required. 
\end{proof}%\end ub on theta_high

Thus, we have
$$\Delta^2\leq \Delta_1 + \frac{1}{C}\leq \frac{\delta^2}{8|G|^6} + \frac{\delta^2}{8|G|^6} = \frac{\delta^2}{4|G|^6},$$
as required.
\end{proof}

\begin{remark}
    We note that the analysis from Section~\ref{sec:high-terms} to bound the $\high$ terms is independent of the dimension of $\alpha$. Therefore, the same analysis would also work to control the $\low$ terms. However, we keep the analysis of bounding the $\low$ terms currently presented in the paper as is since it is much simpler.
\end{remark}

\section{Conclusion}
In this paper, we showed the optimal inapproximability of Max-$k$-\LIN\ over every non-abelian group $G$. The underlying predicate for this constraint satisfaction problem for $k=3$ is given by the following map, $P_G : G^3 \rightarrow \{0,1\}$ where $P_G^{-1}(1)= \{ (a, b, c) \mid a\gp b\gp c = 1_G\}$. The optimal approximation algorithm for this CSP uses the following {\em abelian embedding} property of the predicate into an abelian group. Consider the map $\alpha: G\rightarrow H$ into the abelian group $H:=G/\commutator{G}{G}$ where $\alpha(g)$ is $g\commutator{G}{G}$. For every $(a, b, c)\in P_G^{-1}(1)$, we have $\alpha(a) + \alpha(b) + \alpha(c) = 0_H$ where $H$ is the abelian group $G/\commutator{G}{G}$. Thus, every accepting assignment of the predicate satisfies a linear equation over an abelian group. This property was used towards getting the optimal approximation algorithm with approximation ratio $\frac{1}{\left|\commutator{G}{G}\right|}$ for Max-$3$-\LIN.

Since the publication of this paper, in the subsequent works~\cite{BKMI, BKMII, BKMIV, BKMV}, the notion of abelian embedding of a predicate is used to characterize the optimal approximation ratio of satisfiable CSPs for certain  $3$-ary predicates. We believe that {\em abelian embeddability} is one of the key properties towards getting an optimal approximation algorithm for satisfiable CSPs.

\paragraph{Acknowledgment.} We sincerely thank the reviewer of the JACM who pointed out a mistake in the conference version of the paper. 

\bibliographystyle{alpha}
\bibliography{refs.bib}

\appendix

\section{Proof of Theorem~\ref{theorem:multiLC}}

\newcommand{\NP}{\mathrm{NP}}

\newcommand{\NAE}{\operatorname{NAE}}

\newcommand{\calA}{\mathcal A}

\begin{theorem}[Restatement of Theorem~\ref{theorem:multiLC}]
For any constant parameters $\ell \geq 2, r \in \mathbb{N}$ the following problem is $\np$-hard.
Given an $\ell$-layered Label Cover instance $\calH=(\calV= \{V_1,\ldots, V_\ell\},\{\Pi_{i,j}\}_{1\leq i<j\leq \ell},\{[R_i]\}_{i\in [\ell]})$ where all variable ranges $[R_i]$ are of size $R_i \leq 2^{O(\ell r)}$, distinguish between the following two cases:
\begin{description}
\item[Completeness.] There is an assignment satisfying all the constraints of the Label Cover instance. In this case, we say that $\calH$ is fully satisfiable.
\item[Soundness.] For every $1 \le i < j \le \ell$, no assignment satisfies more than a $2^{-r}$ fraction of the set of constraints $\Pi_{i,j}$ between layers $i$ and $j$. In this case, we say that $\calH$ is at most $2^{-r}$-satisfiable.
\end{description}
Furthermore, the above hardness result also holds for instances with the following additional structure.
\begin{enumerate}
 \item {\em Transitivity of the constraints.} The set of constraints $\Pi_{i, i+1}$ for $1\leq i\leq \ell-1$ define all the other constraints by transitivity. In other words, fix any three layers $1\leq i<j<k\leq \ell$. For any two constraints $\pi_{vu} \in \Pi_{i,j}$ and $\pi_{uw}\in \Pi_{j, k}$, there is a constraint $\pi_{vw}\in \Pi_{i,k}$ such that $\pi_{vw} (a) = \pi_{uw}(\pi_{vu}(a))$ for all $a\in [R_i]$.
 \item {\em Regularity of the constraints.} The constraint graph $G_{i,j}$ is biregular for every $1\leq i<j\leq \ell$. Moreover, the following distributions on constraints are identical:
\begin{enumerate}
\item[(a)] Choose a layer pair $i<j$ uniformly, then a constraint uniformly
  from $\Pi_{i,j}$.
\item[(b)] Choose a full consecutive-layer chain
  $(x_1,\ldots,x_\ell)$ uniformly, independently choose a layer pair
  $i<j$ uniformly, and output its induced constraint $(x_i,x_j)$.
\end{enumerate} 
    \item {\em Regular projection constraints.} There exists a constant $d$ such that $R_{i-1} = d R_{i}$ for all $2\leq i\leq \ell$ and each constraint $\pi_{vu}\in \Pi_{i,j}$ is a $d^{(j-i)}$-to-$1$ projection constraint.

\end{enumerate}

\end{theorem}
\begin{proof}
We use the bounded-occurrence, perfect-completeness gap version of
NAE-$3$SAT as a starting point: there are absolute constants $D\geq1$ and $\eta>0$ for which
it is $\NP$-hard to distinguish satisfiable formulas from formulas of
value at most $1-\eta$, with each clause containing three distinct
variables and each variable occurring between $1$ and $D$ times.
This follows, for example, from the hardness reduction for MAX $3$-SET SPLITTING in~\cite{Guruswami04}. The reduction starts from the constraint satisfaction problem MAXSNE$_4$, whose hardness of approximation is established in~\cite{Hastad2001}, even under a bounded-occurrence restriction. Applying the reduction to these restricted instances yields a NAE-$3$SAT instance with maximum degree bounded by an absolute constant.

We make the instance regular, i.e., every variable appears in the same number of constraints. Let $F$ be the above NAE-$3$SAT instance, with $n$ variables, $m$ clauses, and occurrence
counts $a_x$ for the variable $x$. Take three disjoint copies $F^{(1)},F^{(2)},F^{(3)}$.
For every original variable $x$, add $D-a_x$ copies of the clause
\[
 \NAE(x^{(1)},x^{(2)},x^{(3)}).
\]
Call the resulting formula $F_{\rm reg}$. Each of its $3n$ variables
occurs exactly $D$ times, and its number of clauses is
\[
 3m+\sum_x(D-a_x)=nD.
\]
If $F$ is satisfiable by $\sigma$, use $\sigma$ on the first two copies
and its Boolean complement on the third. NAE clauses are invariant
under complementing all variables, and each added clause has values
$(\sigma(x),\sigma(x),1-\sigma(x))$. Thus $F_{\rm reg}$ is satisfiable.
In the NO case, every assignment violates at least $3\eta m$ clauses
among the three original copies. Since $n\leq3m$, the violated fraction
in $F_{\rm reg}$ is at least
\[
 \frac{3\eta m}{nD}\geq\frac{\eta}{D}.
\]

Form the clause--variable Label Cover instance $G_0=(Y_0,Z_0,E_0)$ of
$F_{\rm reg}$. A left label is one of the six satisfying local
assignments of its NAE clause; a right label is a Boolean value.
Each incidence constraint projects a local assignment to the value
of the corresponding underlying variable. For either prescribed
Boolean value, exactly three of the six satisfying local assignments
extend it. Hence every projection is exactly $3$-to-$1$, including
when the literal is negated.

The graph is simple and $(3,D)$-biregular. Completeness is perfect.
In the NO case, any assignment to the right vertices violates at least
an $\eta/D$ fraction of the clauses. At each such clause, every
satisfying local assignment disagrees with the right assignment on
at least one of its three incidences. Therefore
\[
 \val(G_0)\leq1-\gamma,\qquad \gamma:=\frac{\eta}{3D}>0.
\]

Apply $t$ independent parallel repetitions. Since the initial
alphabets and gap are absolute constants, Raz's theorem~\cite{Raz1998}
gives an absolute $c>0$ such that the repeated value is at most
$2^{-ct}$ in the NO case. Choose $t=O(r)$ large enough that this is
at most $2^{-r}$. Denote the repeated instance by $G=(Y,Z,E)$ and
its alphabets by
\[
 A=[6]^t,\qquad B=\{0,1\}^t.
\]
Parallel repetition preserves the biregularity:
the two degrees are
\[
 \Delta_Y=3^t,\qquad \Delta_Z=D^t.
\]
Every repeated projection $\rho_{yz}:A\to B$ is exactly
 $d:=3^t$
to one.
Consequently $|A|=d|B|$, both alphabets have size $2^{O(r)}$,
completeness is perfect, and the NO-case value is at most $2^{-r}$.
This establishes the required small-soundness, biregular base Label Cover. 

We now reduce the Label Cover instance to a Layered Label Cover instance.

\noindent\emph{The Layered Label Cover construction.}
Following the product construction of~\cite{DGKR05},
put
\[
 V_i=Z^i\times Y^{\ell-i},\qquad
 \calA_i=B^i\times A^{\ell-i},\qquad
 R_i=|\calA_i|.
\]
Fix bijections between $\calA_i$ and $[R_i]$.
For $i<j$, include a constraint from $v\in V_i$ to $u\in V_j$
precisely when
\[
 \begin{cases}
 v_p=u_p,&p\leq i\text{ or }p>j,\\
 (v_p,u_p)\in E,&i<p\leq j.
 \end{cases}
\]
For $a=(a_1,\ldots,a_\ell)\in\calA_i$, define
\[
 \bigl(\pi_{vu}(a)\bigr)_p=
 \begin{cases}
  \rho_{v_pu_p}(a_p),&i<p\leq j,\\
  a_p,&p\leq i\text{ or }p>j.
 \end{cases}
 \label{eq:layered-projection}
\]

The completeness and soundness of the construction are given in~\cite{DGKR05}. Here, we verify the structural properties mentioned in the theorem statement.

\begin{enumerate}
\item Suppose $(v,u)\in\Pi_{i,j}$ and $(u,w)\in\Pi_{j,k}$.
For $i<p\leq j$, $u_p=w_p$; for $j<p\leq k$, $v_p=u_p$;
outside $(i,k]$, all three coordinates agree. The required base
edges therefore exist throughout $(i,k]$, so $(v,w)\in\Pi_{i,k}$.
The two maps change disjoint
coordinate intervals, giving
$\pi_{vw}=\pi_{uw}\circ\pi_{vu}$.
Conversely, an edge $(v,w)\in\Pi_{i,k}$ has the unique intermediate
vertex
\[
 u=(w_1,\ldots,w_j,v_{j+1},\ldots,v_\ell)\in V_j.
\]
Iterating gives its unique consecutive-layer factorization.

\item For a fixed $v\in V_i$, its neighbors in $V_j$ are obtained
by independently choosing a base neighbor for each of
$v_{i+1},\ldots,v_j$. Its degree is therefore
$\Delta_Y^{j-i}$. In the reverse direction, every $u\in V_j$
has degree $\Delta_Z^{j-i}$. Hence $G_{i,j}$ is biregular.

\emph{Equality of the two sampling distributions.}
Every full chain is uniquely specified by
\[
 z_1\in Z,\qquad (y_p,z_p)\in E\quad(2\leq p\leq\ell),
\]
through
\[
 x_i=(z_1,\ldots,z_i,y_{i+1},\ldots,y_\ell).
\]
Thus independent uniform choices of these variables give a uniform full
chain. A uniform edge of the biregular base graph has uniform
marginals on $Y$ and $Z$. For each fixed $i<j$, the shared prefix,
intervening base edges, and shared suffix of $(x_i,x_j)$ are therefore
independent and uniform on
\[
 Z^i,\qquad E^{j-i},\qquad Y^{\ell-j},
\]
respectively. Using the fact that the constraints between these layers are in bijection with
\[
 Z^i\times E^{j-i}\times Y^{\ell-j},
\]
for every
$(v,u)\in\Pi_{i,j}$ we obtain
\[
 \Pr[(x_i,x_j)=(v,u)]
 =\frac{1}{|Z|^i|E|^{j-i}|Y|^{\ell-j}}
 =\frac{1}{|\Pi_{i,j}|}.
\]
An independent uniform layer pair therefore gives each constraint
in $\Pi_{i,j}$ probability
$\binom{\ell}{2}^{-1}|\Pi_{i,j}|^{-1}$ under both procedures.

\emph{Alphabet sizes and reduction size.}
For every layer,
\[
 R_i=2^{ti}6^{t(\ell-i)}
 \leq6^{t\ell}=2^{O(\ell r)}.
\]
All operations are deterministic and have polynomial size for fixed
$\ell,r$: exact degree regularization has linear size, repetition
uses $t=O(r)$ coordinates, and layering uses $\ell$ coordinates.
The integer $d=3^t$ depends only on the fixed soundness parameter
$r$. 

\item Since
\[
 R_i=|B|^i|A|^{\ell-i},\qquad |A|=d|B|,
\]
we have $R_{i-1}=dR_i$. For any target label
$b=(b_1,\ldots,b_\ell)\in\calA_j$, the unchanged coordinates
of a preimage under $\pi_{vu}$ are forced, and each changed
coordinate has exactly $d$ choices. Thus
\[
 |\pi_{vu}^{-1}(b)|
 =\prod_{p=i+1}^{j}|\rho_{v_pu_p}^{-1}(b_p)|
 =d^{j-i}.
\]
\end{enumerate}
This proves the theorem and the sampling assertion.
\end{proof}

\end{document}